\documentclass[final]{IEEEtran}
\IEEEoverridecommandlockouts
\usepackage{cite}
\usepackage{amsmath,amssymb,amsfonts,amsthm}
\usepackage{algorithmic}
\usepackage{graphicx}
\usepackage{textcomp}
\usepackage{xcolor}
\usepackage{color}
\usepackage{algorithmic}
\usepackage{url}
\usepackage{lipsum}
\usepackage{multirow}
\usepackage[normalem]{ulem}
\usepackage{afterpage}
\usepackage{bbm}
\usepackage{dsfont}
\usepackage[ruled,lined,boxed]{algorithm2e}
\usepackage{mathtools, nccmath}
\usepackage{lipsum}
\usepackage{bm}
\usepackage{threeparttable}
\usepackage[colorlinks=true,citecolor=blue,linkcolor=blue]{hyperref}

\ifCLASSOPTIONcompsoc
\usepackage[caption=false, font=normalsize, labelfont=sf, textfont=sf]{subfig}
\else
\usepackage[caption=false, font=footnotesize]{subfig}
\fi
\DeclarePairedDelimiter{\nint}\lfloor\rceil
\DeclareMathAlphabet{\mathbcal}{OMS}{cmsy}{b}{n}
\def\BibTeX{{\rm B\kern-.05em{\sc i\kern-.025em b}\kern-.08em
		T\kern-.1667em\lower.7ex\hbox{E}\kern-.125emX}}

\newtheorem{prop}{Proposition}

\newtheoremstyle{iremark}
{\topsep}   
{\topsep}   
{\upshape}  
{0pt}       
{\itshape}  
{.}         
{5pt plus 1pt minus 1pt} 
{\thmname{#1}\thmnumber{ \itshape#2}\thmnote{ (#3)}} 
\theoremstyle{iremark}
\newtheorem{remark}{Remark}

\mathchardef\mhyphen="2D
\makeatletter 
\let\myorg@bibitem\bibitem
\def\bibitem#1#2\par{%
	\@ifundefined{bibitem@#1}{%
		\myorg@bibitem{#1}#2\par
	}{%
		\begingroup
		\color{\csname bibitem@#1\endcsname}%
		\myorg@bibitem{#1}#2\par
		\endgroup
	}%
}
\makeatother
\begin{document}
\title{Sensing-Enhanced Channel Estimation for Near-Field XL-MIMO Systems\\
}
\author{Shicong~Liu,~\IEEEmembership{Graduate Student Member,~IEEE}, Xianghao~Yu,~\IEEEmembership{Senior Member,~IEEE}, Zhen~Gao,~\IEEEmembership{Member,~IEEE}, Jie~Xu,~\IEEEmembership{Senior Member,~IEEE}, Derrick~Wing~Kwan~Ng,~\IEEEmembership{Fellow,~IEEE}, and Shuguang~Cui,~\IEEEmembership{Fellow,~IEEE}
	\vspace{-1mm}
	\thanks{		
		Part of this paper was presented at IEEE International Conference on Communications (ICC), Denver, CO, USA, June, 2024~\cite{ICC24Liu2}. ({\em Corresponding author: Xianghao Yu.})
		
		Shicong Liu and Xianghao Yu are with the Department of Electrical Engineering, City University of Hong Kong, Hong Kong (email: sc.liu@my.cityu.edu.hk, alex.yu@cityu.edu.hk).
		
		
		Zhen Gao is with the Advanced Technology Research Institute, BIT (Jinan), Jinan 250307, China, and also with the Yangtze Delta Region Academy, BIT (Jiaxing), Jiaxing 314019, China, also with the Beijing Institute of Technology (Zhuhai), Zhuhai 519088, China (e-mail: gaozhen16@bit.edu.cn).
		
		Jie Xu and Shuguang Cui are with the School of Science and Engineering (SSE), the Shenzhen Future Network of Intelligence Institute (FNii-Shenzhen), and the Guangdong Provincial Key Laboratory of Future Networks of Intelligence, The Chinese University of Hong Kong (Shenzhen), Guangdong 518172, China (email: xujie@cuhk.edu.cn, shuguangcui@cuhk.edu.cn).
		
		Derrick Wing Kwan Ng is with the School of Electrical Engineering and Telecommunications, University of New South Wales, Sydney, Australia (email: w.k.ng@unsw.edu.au).
		
		\texttt{https://github.com/scliubit/sensing-ce-xlmimo}
	}
}
\maketitle
\begin{abstract}
	Future sixth-generation (6G) systems are expected to leverage extremely large-scale multiple-input multiple-output (XL-MIMO) technology, which significantly expands the range of the near-field region. 
	The spherical wavefront characteristics in the near field introduce additional degrees of freedom (DoFs), namely distance and angle, into the channel model, which leads to unique challenges in channel estimation (CE). 
	In this paper, we propose a new sensing-enhanced uplink CE scheme for near-field XL-MIMO, which notably reduces the required quantity of \textit{baseband samples} and the \textit{dictionary size}.
	In particular, we first propose a sensing method that can be accomplished in a single time slot. 
	It employs power sensors embedded within the antenna elements to measure the received power pattern rather than baseband samples. A time inversion algorithm is then proposed to precisely estimate the locations of users and scatterers, which offers a substantially lower computational complexity. 
	Based on the estimated locations from sensing, a novel dictionary is then proposed by considering the eigen-problem based on the near-field transmission model, which facilitates efficient near-field CE with less baseband sampling and a more lightweight dictionary. 
	Moreover, we derive the general form of the eigenvectors associated with the near-field channel matrix, revealing their noteworthy connection to the discrete prolate spheroidal sequence (DPSS). 
	Simulation results unveil that the proposed time inversion algorithm achieves accurate localization with power measurements only, and remarkably outperforms various widely-adopted algorithms in terms of computational complexity. 
	Furthermore, the proposed eigen-dictionary considerably improves the accuracy in CE with a compact dictionary size and a drastic reduction in baseband samples by up to 66\%.
\end{abstract}

\begin{IEEEkeywords}
	Channel estimation, compressive sensing, discrete prolate spheroidal sequence, near-field localization, sensing-enhanced communication.
\end{IEEEkeywords}

\section{Introduction}
\bstctlcite{IEEEexample:BSTcontrol}
\IEEEPARstart{T}{he} development of massive multiple-input multiple-output (MIMO) systems has spurred a vision to reshape and control transmission environments of propagating electromagnetic (EM) waves. 
This vision leads to the emergence of advanced technologies, such as cell-free massive MIMO and reconfigurable intelligent surfaces (RIS), that enhance service coverage and eliminate dead zones in wireless networks\cite{9690635,7827017}. Particularly, for centralized large-scale antenna array deployment strategies, like RIS and extremely large-scale MIMO (XL-MIMO)\cite{10098681}, their vast apertures drastically expand the boundaries of the near-field region~\cite{9139337}. While this transition in XL-MIMO systems enables efficient beam focusing for achieving 
higher transmission rates~\cite{9738442}, it introduces novel challenges in the design of wireless systems.

In the practical near-field region, the spherical wavefront of electromagnetic waves diverges from the planar wavefront typically assumed in the far-field counterpart. 
This introduces two degrees of freedom (DoFs), i.e., angle and distance, to accurately describe the spatial characteristics, thereby complicating the task of channel estimation (CE)~\cite{9620081}. 
Moreover, the proliferation of antennas in XL-MIMO arrays not only substantially enlarges the dimensionality of channel matrices to be estimated, but also makes it impractical to deploy a dedicated radio frequency (RF) chain for each antenna element~\cite{6717211}. Therefore, hybrid analog and digital antenna array architecture  with a limited number of RF chains has been widely adopted~\cite{7400949}. Nevertheless, this solution only allows a low-dimensional baseband sample per time slot, which poses additional challenges in accurately reconstructing the high-dimensional channel in XL-MIMO systems. 

\subsection{Related Works}

To mitigate this challenge, 
compressive sensing (CS)-based techniques~\cite{10153711,9547795,8949454,9693928,10195974,10217152} have been proposed for XL-MIMO CE by leveraging the inherent sparsity of wireless channel matrices~\cite{1614066}. 
Particularly, iterative greedy CS-based methods~\cite{9693928,8949454,10195974,10217152} have been extensively investigated in the CE literature due to their high computational efficiency, which are typically developed based on heuristic selection rules by selecting matched codewords in the dictionary.
Unfortunately, the conventional far-field sparse representation dictionary, e.g., discrete Fourier transform (DFT) dictionary, no longer matches the near-field channel model~\cite{8949454}. 
As a remedy, a spherical wave dictionary was introduced in~\cite{9693928}, based on which a polar-domain sampling scheme was tailored for near-field CE. 
Subsequently, a hierarchical near-field dictionary was proposed, where the upper-layer dictionaries are exploited for target location search while the lower-layer ones are adopted to achieve the highest beam gain around the steering points~\cite{10195974}. Besides, an alternative dictionary design was recently presented in~\cite{10217152} by utilizing the spatial-chirp beam. 

However, due to the two DoFs brought by the non-negligible spherical wave transmission effect, existing greedy CS-based algorithms encounter two common demerits. First, additional codewords are required in the dictionary to represent the more complicated channel with two DoFs. This results in an oversized dictionary in~\cite{8949454,9693928,10195974}, and therefore imposes a more stringent requirement for dictionary storage. Second, and more importantly, 
the additional DoF in the near-field channel model under XL-MIMO systems introduces a higher rank to the channel matrix~\cite{9848802}, which increases the sparsity ratio of near-field channels. 
In this case, CS-based algorithms necessitate a larger number of pilots to obtain more observations of the less sparse near-field channel, thereby achieving a satisfactory reconstruction accuracy. 
Therefore, while the number of baseband samples per time slot is shrunk under the hybrid antenna array architecture, the drastically increased pilot length induces a prohibitively huge number of baseband samples throughout the CE procedure~\cite{9693928}. 
In summary, designing a \textit{lightweight} dictionary and achieving \textit{minimal baseband samples} in CE for near-field XL-MIMO remain open research problems.

Notably, an ideal dictionary needs to conform to the spatial characteristics of the channel in both distance and angle domains. This stipulation implies that acquiring localization information for both user equipments (UEs) and scatterers is a fundamental task in dictionary design for near-field CE~\cite{10273424,9913211}. Specifically, when \textit{prior localization information} pertaining to the propagation environment is incorporated into the dictionary, it can significantly reduce the need for further pilots and subsequent baseband samples for near-field CE. 
In recent years, integrated sensing and communication (ISAC) has made significant progress, posing sensing as a promising functionality in future communication systems~\cite{10188491}.  
Thanks to the large aperture provided by XL-MIMO arrays, efficient location estimation~\cite{9707730} as well as movement tracking~\cite{9508850} of UEs and scatterers become possible. 
This motivates us to customize a novel sensing-enhanced CE scheme for near-field XL-MIMO.

There have been some initial attempts on localization via XL-MIMO or sensing-assisted CE. 
Near-field radar sensing utilizing orthogonal frequency division multiplexing (OFDM) waveforms was validated in~\cite{9707730}, while~\cite{9508850} proposed a near-field tracking method by measuring the curvature of arrival of the spherical wavefront from the UE. Besides, a localization method based on the multiple signal classification (MUSIC) algorithm has also been proposed to decouple the estimation of angles of arrival (AoAs) and distances of UEs located in the near field~\cite{10149471}. 
However, an excessive number of additional baseband samples are required for localization when directly integrating such methods into existing CE schemes, which counteracts the reduction of samples for CE brought by localization. 
{Furthermore, a ping-pong pilot-based active-sensing method was proposed to achieve near-field beam alignment between the base station (BS) and the UE\cite{yuanwei}, where neural networks are utilized for beamforming optimization. Nevertheless, the pilot overhead and synchronization issues inherent in the alternating uplink and downlink training procedure remain challenging.} 
To further improve the utilization of pilot signals, a sensing-assisted CE scheme was then proposed in~\cite{ren2023sensingassisted}. While the pilot signals are reused for both localization and CE procedures and therefore no baseband sampling is required for localization, fully-digital XL-MIMO arrays are essentially the prerequisite of existing sensing-assisted CE. 
This array architecture requires a large number of baseband samples (typically dozens of times the number of antennas on the XL-MIMO array) via power-hungry analog-to-digital converters (ADCs) and, hence, is far from practical. 
Furthermore, the computational complexity of  classical localization algorithms, e.g., MUSIC, applied in~\cite{10149471} and \cite{ren2023sensingassisted} is exceptionally high, which is caused by massive matrix multiplications for super-resolution grid search on both distance and angle domains. 
Therefore, pragmatic CE schemes enhanced by localization with \textit{smaller sampling overhead} and \textit{lower computational complexity} still require further exploration.

More recently, a localization method based on lens antenna arrays~\cite{9314267} was proposed to alleviate the requirement of baseband samples by utilizing a novel hardware architecture, while the extremely large focal arc of the lens antenna array renders it difficult to deploy in practice. 
Despite these challenges, the expanded aperture in near-field wireless systems inspires the development of novel hardware architectures, which should aim to effectively utilize XL-MIMO for localization tasks, thereby significantly enhancing near-field CE capabilities.

\subsection{Contributions}
In this paper, we propose a sensing-enhanced CE scheme in a near-field XL-MIMO system operating in  time division duplexing (TDD) mode, whereas our contributions are summarized as follows:
\begin{itemize}
	\item We introduce a cost-effective XL-MIMO transceiver architecture, implemented with off-the-shelf power sensors, to measure the power of the received EM waves. This architecture not only enhances the sensing capability of near-field propagation environment but also ensures a practical and efficient implementation in real-world scenarios.
	\item We propose a time inversion algorithm to estimate the locations of both UEs and scatterers. Unlike the conventional localization methods~\cite{ren2023sensingassisted,8949454,10149471}, the proposed algorithm does not require any baseband samples and can be implemented by computationally-efficient fast Fourier transforms (FFTs).
	\item We exploit the estimated localization coordinates in the CE procedure to derive a novel dictionary, by exploring the eigenvalue-decomposition (EVD) of the near-field channel matrix. 
	In particular, the column-wise orthogonality between eigenvectors inherently reduces the dictionary size to a minimum. We also reveal that the codewords in the proposed dictionary admit the form of discrete prolate spheroidal sequences (DPSS).
	\item Numerical results demonstrate that the proposed time inversion algorithm achieves accurate localization without baseband sampling. In addition, the proposed DPSS-based dictionary incorporating the estimated location coordinates of the UE and scatterers achieves satisfactory accuracy while requiring a number of baseband samples that is up to 77\% smaller. 
	Besides, the required size of the proposed DPSS-based dictionary is shrunk by up to 88\% compared to the conventional DFT and spherical wave dictionaries, which leads to less stringent storage requirements. 
\end{itemize}

\subsection{Notations}

We adopt normal-face letters to denote scalars and lowercase (uppercase) boldface letters to denote column vectors (matrices). The $k$-th row vector and the $m$-th column vector of matrix ${\bf H}\in\mathbb{C}^{K\times M}$ are denoted as ${\bf H}[{k,:}]$ and ${\bf H}[{:,m}]$, respectively. The operator $\odot$, $\otimes$, and $\circledast$ denote the Hadamard product, Kronecker product, and linear convolution operators, respectively. $\{{\bf H}_n\}_{n=1}^N$ denotes a matrix set with the cardinality of $N$. ${\rm vec}(\cdot)$ denotes the vectorization operator, while ${\rm diag}({\bf h})$ constructs a diagonal matrix whose diagonal elements are extracted from ${\bf h}$. The superscripts $(\cdot)^{T}$, $(\cdot)^{\rm *}$, $(\cdot)^{H}$, and $(\cdot)^{\dagger}$ represent the transpose, conjugate, conjugate transpose, and pseudo-inverse operators, respectively. We also exploit $\lceil{\cdot}\rceil$, $\lfloor{\cdot}\rfloor$, and $\nint{\cdot}$ to denote the ceiling, flooring, and rounding operators, respectively. $\mathcal{CN}(\mu,\sigma^2)$ and $\mathcal{U}[a,b]$ represent the complex Gaussian distribution with mean $\mu$ and variance $\sigma^2$, and the uniform distribution over $[a,b]$, respectively. $\mathbb{E}[\cdot]$ denotes the statistical expectation operator. The $\ell_0$-norm of a vector, $\Vert\cdot\Vert_0$, counts the number of its non-zero elements. We use ${\rm Re}(z)$ to denote the real part of a complex number, and represent the imaginary unit as $\jmath$ with $\jmath^2=-1$. $\mathcal{O}(\cdot)$ denotes the upper bound on the worst-case number of multiplications required by an algorithm, while $\propto$ denotes that one quantity varies directly in proportion to another.

\section{System Model}
In this section, we firstly describe the considered near-field transmission scenario and present the near-field channel model. Furthermore, the uplink sensing-enhanced CE protocol is illustrated. The power sensor-based and hybrid antenna array is then introduced, 
whereas the corresponding uplink sensing and training models are further provided.
\subsection{Channel Model}
\begin{figure}[t]
	\centering
	\includegraphics[width=0.45\textwidth]{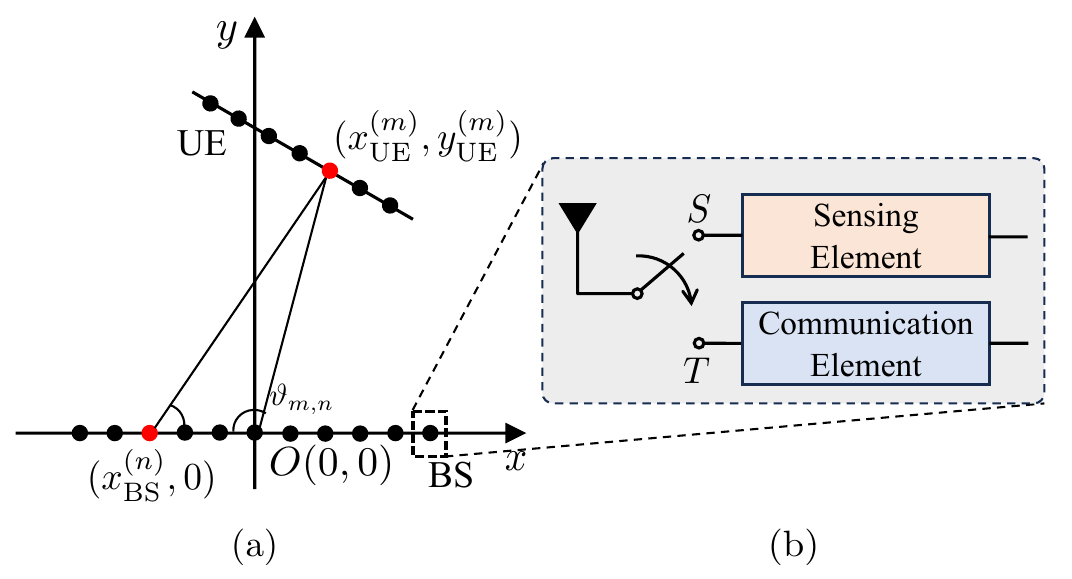}
	\caption{(a) The considered near-field transmission scenario and (b) the proposed dual-functional architecture for sensing-enhanced communication. RF switches are adopted on antenna elements in (b), which selectively activate either the $S$ branch for sensing or the $T$ branch for training.}
	\label{fig:sysmodela}
\end{figure}
Consider a scenario with a UE array\footnote{The proposed design in this paper can be extended to multi-user scenarios by assigning orthogonal pilots for different UEs. } and a BS equipped with a hybrid analog-digital architecture XL-MIMO array operating in the TDD mode within its near-field region, where both the UE and BS are equipped with uniform linear arrays (ULAs)\footnote{We consider ULA here for brevity, while it can be generalized to other antenna geometries. For example, it can be extended to uniform planar array (UPA) by applying Kronecker products to steering vectors in~\eqref{eq:LoS}.}, as shown in Fig.~\ref{fig:sysmodela}(a). The Cartesian coordinates of the $m$-th element of the UE's antenna array and the $n$-th element of the BS antenna array are denoted by ${\bf r}_{\rm UE}^{(m)}=(x_{\rm UE}^{(m)},y_{\rm UE}^{(m)})^T$ and ${\bf r}_{\rm BS}^{(n)}=(x_{\rm BS}^{(n)},0)^T$, respectively. For uni-polarized antenna elements, the spherical wave spatial impulse response in the scalar form~\cite{9723331} is given by
\begin{equation}
	{g}\left({\bf r}_{\rm BS}^{(n)},{\bf r}_{\rm UE}^{(m)}\right) = \frac{e^{-\jmath\kappa \Vert {\bf r}_{\rm UE}^{(m)}-{\bf r}_{\rm BS}^{(n)} \Vert}}{\Vert {\bf r}_{\rm UE}^{(m)}-{\bf r}_{\rm BS}^{(n)} \Vert},
	\label{eq:sv}
\end{equation}
where $\kappa = 2\pi/\lambda$ is the wavenumber, and $\lambda$ is the wavelength. For the $m$-th ($1\leq m\leq N_{\rm UE}$) antenna element of the UE array, the uplink line-of-sight (LoS) wireless channel is modeled by using the near-field steering vector as
\begin{equation}
	\begin{aligned}
		\mathbf{H}_{\rm LoS}[:,m] &= {\mathbf g}_{\rm BS}\left({\bf r}_{\rm UE}^{(m)}\right)\\&= \left[ {g}\left({\bf r}_{\rm UE}^{(m)},{\bf r}_{\rm BS}^{(1)}\right),\cdots, {g}\left({\bf r}_{\rm UE}^{(m)},{\bf r}_{\rm BS}^{(N_{\rm BS})}\right) \right]^{T},
	\end{aligned}
	\label{eq:LoS}
\end{equation}
where ${\mathbf g}_{\rm BS}(\cdot)$ is the near-field steering vector from the $m$-th antenna of the UE array to the BS array, and $N_{\rm BS}$ and $N_{\rm UE}$ denote the numbers of antennas at the BS and UE, respectively. Considering multipath fading channel model~\cite{10220205}, the non-LoS (NLoS) components are modeled as
\begin{equation}
	{\bf H}_{\rm NLoS} = \sum_{\ell=1}^{L} \alpha_\ell {\bf g}_{\rm BS}\left( {\bf r}_{\rm S}^{(\ell)} \right) {\bf g}^{T}_{\rm UE}\left( {\bf r}_{\rm S}^{(\ell)} \right),
	\label{eq:NLoS}
\end{equation}
where ${\bf g}_{\rm UE}\left(\cdot \right)$ has a similar form to ${\bf g}_{\rm BS}\left(\cdot \right)$ in~\eqref{eq:LoS}, ${\bf r}_{\rm S}^{(\ell)}=({x}_{\rm S}^{(\ell)},{y}_{\rm S}^{(\ell)})^T$ denotes the location of the $\ell$-th scatterer, $L$ denotes the number of multipaths, and $\alpha_\ell\sim\mathcal{CN}(0,1/L)$ for all $1\leq \ell\leq L$ denotes the complex channel attenuation. The overall uplink channel model is then expressed as
\begin{equation}
	{\bf H} = \mathbf{H}_{\rm LoS} + \mathbf{H}_{\rm NLoS}.
	\label{eq:channelmodel}
\end{equation}

From \eqref{eq:sv}, \eqref{eq:LoS}, and Fig.~\ref{fig:sysmodela}(a), it can be determined that each element in the near-field steering vector is given by
\begin{equation}
	g\left({\bf r}_{\rm BS}^{(n)},{\bf r}_{\rm UE}^{(m)}\right) =\frac{e^{-\jmath \kappa\sqrt{\Vert{\bf r}_{\rm BS}^{(n)}\Vert^2+\Vert{\bf r}_{\rm UE}^{(m)}\Vert^2-2\Vert{\bf r}_{\rm BS}^{(n)}\Vert\Vert{\bf r}_{\rm UE}^{(m)}\Vert\cos(\vartheta_{m,n})}}}{\Vert{\bf r}_{\rm BS}^{(n)}-{\bf r}_{\rm UE}^{(m)}\Vert},
\end{equation}
where the information in both distance and angular domains is required. Note that this is the main difference between the near-field channel model and the conventional far-field counterpart, where only angular information is decisive~\cite{10217152}. Hence, the inclusion of these additional parameters related to distance and its coupling with the angular information introduces heightened complexity and difficulty in CE problems.
\subsection{Sensing-Enhanced CE Protocol and Power Senor-based XL-MIMO}
\label{sec:arch}
We consider the uplink CE at the BS, which can be used to facilitate the downlink transmission. Conventionally, 
the UE transmits $\tau$ uplink pilot signals during the training stage, while the BS equipped with hybrid analog-digital architecture estimates the channel based on the received pilots~\cite{6756990,9693928}. However, the pilot length $\tau$ and the resulting number of baseband samples at the BS can be exceedingly large due to limited RF chains in the hybrid array architecture and the complicated near-field channel model~\cite{10273424,10143629}. 

In this paper, we propose a new sensing-enhanced CE protocol, where a single sensing time slot is inserted before the conventional pilot signal transmission, as shown in Fig.~\ref{fig:protocol}. {In particular, in this additional slot, the RF switch is set to the $S$ branch (cf. Fig.~\ref{fig:sysmodela}(b)), enabling the UE to transmit an uplink sensing signal to realize the localization of both the UE and scatterers. In the subsequent $\tau$ slots, the RF switch is then set to the $T$ branch, where the UE transmits uplink pilots to facilitate CE.} The main benefit of adding such a sensing time slot is that the estimated location coordinates can provide accurate \textit{prior localization information}, thereby facilitating efficient near-field CE. {This approach effectively reduces the pilot length $\tau$ in the subsequent training stage}.

\begin{figure}[t]
	\centering
	\includegraphics[width=0.275\textwidth]{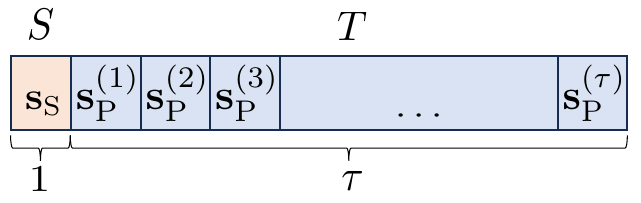}
	\caption{Proposed sensing-enhanced uplink CE protocol. One sensing signal $\mathbf{s}_\mathrm{S}$ is transmitted in the sensing ($S$) stage and $\tau$ pilot signals $\mathbf{s}_\mathrm{P}^{(t)}$ are transmitted subsequently in the training ($T$) stage.}
	\label{fig:protocol}
	\vspace{-3mm}
\end{figure}
\begin{figure}[t]
	\centering
	\includegraphics[width=0.35\textwidth]{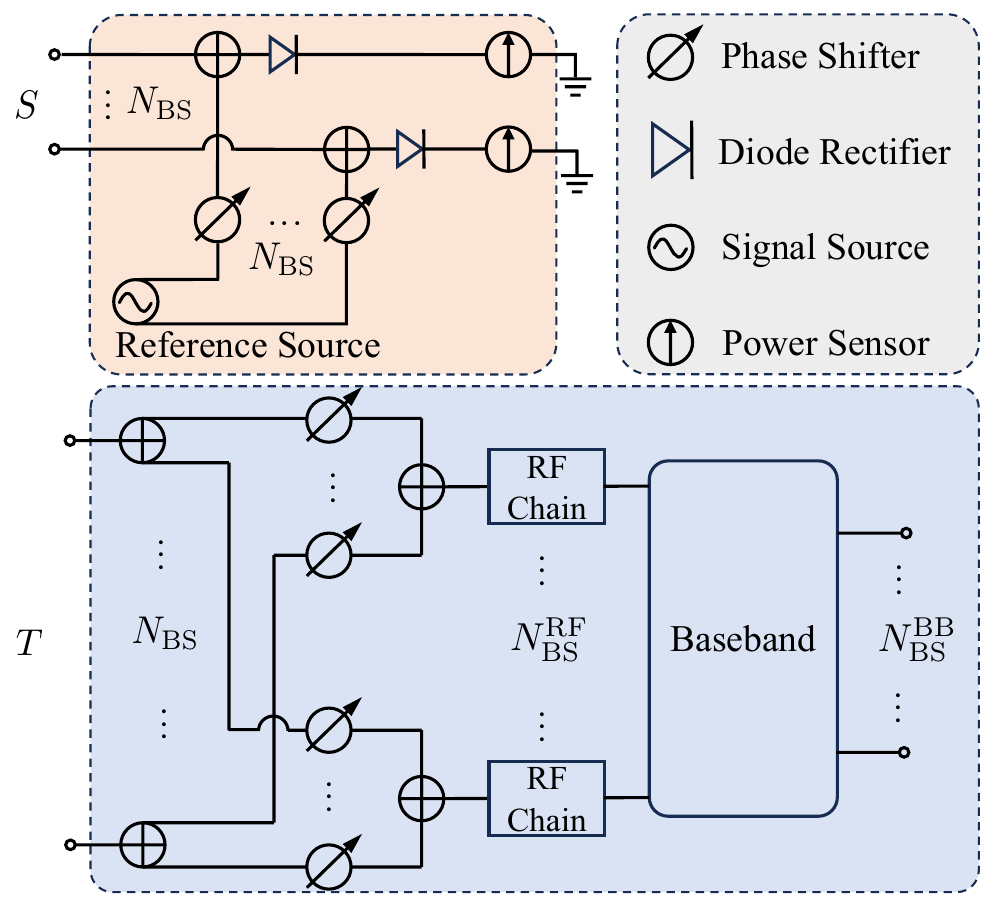}
	\caption{The detailed hardware schematic of antenna elements at the BS array~\cite{10042005,942570}.}
	\label{fig:sysmodelb}
\end{figure}

{However, it should be noted that conventional localization techniques generally necessitate extra baseband samples that are often several times the size of antenna arrays~\cite{ren2023sensingassisted}. In fact, each baseband sample is generated through an RF front end with a high-precision ADC, whose power consumption increases exponentially with the number of quantization bits, and can reach up to the watt level in XL-MIMO systems~\cite{4684631}. This forms the first obstacle for practical localization methods in the baseband. Moreover, in the widely considered hybrid analog-digital array architecture~\cite{6717211}, the number of RF chains and the corresponding ADCs is much fewer than that of antenna elements. Therefore, more time slots are needed to obtain a sufficient number of baseband measurements for accurate localization, which leaves significantly fewer time slots available for subsequent CE.

To address this problem, in this paper we propose a power sensor-based localization scheme in the RF domain. It can be accomplished in \textit{only one} sensing slot and promote efficient near-field CE.} As shown in Fig.~\ref{fig:sysmodelb}, the sensing part, composed of off-the-shelf hardware components including phase shifters, combiners, local oscillator, and power sensors, distinguishing it from other setups~\cite{MaHongGaoJingLiuBaiChengCui}. 
{Rather than relying on energy-consuming and costly high-frequency ADCs and baseband samplers, this approach utilizes power sensors to measure the power pattern of the received wave combined with the phase-shifted reference wave in an energy-efficient (down to milliwatt level) and economical manner~\cite{9769720,7448359}. }

A similar architecture has been extensively utilized in holographic transceivers for microwave sensing~\cite{942570,tagoram} and imaging~\cite{9246248} within the near-field region, primarily due to its low-cost and plug-and-play characteristics. 
Recently, an initial attempt to incorporate power sensor-based antenna elements into wireless communications was made in~\cite{10042005} for RIS-assisted systems, where efficient parameter estimation and beamforming algorithms were developed. In this paper, we shall also leverage this practical hardware architecture, based on which baseband sampling-free localization is designed.

\subsection{Sensing Model}
This subsection considers the first time slot when the RF switches equipped at the BS are set to the $S$ branch, such that the sensing function is activated. Based on the uplink sensing signal transmitted from the UE, the received object wave at the $n$-th antenna element of the BS is described as
\begin{align}
	{\bf e}_{\rm obj}[n] &\!=\!\sqrt{\frac{1}{N_{\rm UE}}}\sum_{m=1}^{N_{\rm UE}} {g}\!\left( {\bf r}_{\rm UE}^{(m)},{\bf r}_{\rm BS}^{(n)} \right){\bf s}_{\rm S}[m] \label{eq:obj}\\&{~~~}\!+\! \sum_{\ell=1}^{L}\!\sum_{m=1}^{N_{\rm UE}}\!\!\alpha_\ell {g}\!\left( {\bf r}_{\rm S}^{(\ell)},{\bf r}_{\rm BS}^{(n)} \right)\! {g}\!\left( {\bf r}_{\rm S}^{(\ell)},{\bf r}_{\rm UE}^{(m)} \right){\bf s}_{\rm S}[m]\!+\!{\bf e}_{\rm n}[n]\notag,
\end{align}
where ${\bf s}_{\rm S}\in\mathbb{C}^{N_{\rm UE}}$ denotes the sensing signal, and ${\bf e}_{\rm n}\in\mathbb{C}^{N_{\rm BS}}$ is the additive white Gaussian noise (AWGN) vector. 
Without loss of generality, we consider the case when ${\bf s}_{\rm S}$ is an all-one vector~\cite{9314267}. 
As depicted in Fig.~\ref{fig:sysmodelb}, the received object waves are then combined with the reference wave as ${\bf e}_{\rm ref}[n] + {\bf e}_{\rm obj}[n]$, where ${\bf e}_{\rm ref} = A [e^{\jmath\xi_1},\cdots,e^{\jmath\xi_{N_{\rm BS}}}]^T$ is the phase-shifted signal from the local oscillator with $A$ denoting the manipulated amplitude of the reference wave and $\xi_{n}$ denoting the phase shift at the $n$-th antenna element of the BS array. Then, the $n$-th power sensor~\cite{10042005,942570} deployed at the XL-MIMO transceiver measures the power pattern of combined signals as
\begin{equation}
	\begin{aligned}
		{\bf p}[n] &= \left\vert  {\bf e}_{\rm ref}[n]+{\bf e}_{\rm obj}[n] \right\vert^2 \\&= \left\vert  {\bf e}_{\rm ref}[n] \right\vert^2+\left\vert {\bf e}_{\rm obj}[n] \right\vert^2 +{\bf e}_{\rm ref}^{*}[n]{\bf e}_{\rm obj}[n] \\&~~~+ {\bf e}_{\rm ref}[n]{\bf e}^{*}_{\rm obj}[n],
	\end{aligned}
	\label{eq:pwrrecord}
\end{equation}
from which we desire to recover the location information $\{{\bf r}_{\rm S}^{(\ell)} \}_{\ell=1}^L$ in ${\bf e}_{\rm obj}[n]$.
\subsection{Pilot Training Model}
Next, we consider the subsequent uplink pilot training stage, when the UE transmits the remaining pilot signals ${\bf s}_{\rm P}^{(t)}\in\mathbb{C}^{N_{\rm UE}^{\rm S}}$ ($1\!\leq\! t\!\leq\! \tau$) for CE. 
To reduce the number of pricey and power-hungry RF chains in XL-MIMO systems, we consider hybrid transceiver architectures~\cite{9693928,10195974,10217152} for both the BS and UE. 
In this regard, the received signal at the BS in the $t$-th time slot is expressed as
\begin{equation}
	{\bf y}^{(t)} = \left( {\bf W}_{\rm RF}^{(t)}{\bf W}_{\rm BB}^{(t)} \right)^{H}\left( {\bf H} {\bf F}_{\rm RF}^{(t)}{\bf F}_{\rm BB}^{(t)}{\bf s}_{\rm P}^{(t)} + {\bf n}^{(t)} \right),
	\label{eq:dltrain}
\end{equation}
where ${\bf W}_{\rm RF}^{(t)}\in\mathbb{C}^{N_{\rm BS}\times N^{\rm RF}_{\rm BS}}$ and ${\bf W}_{\rm BB}^{(t)}\in\mathbb{C}^{N^{\rm RF}_{\rm BS}\times N^{\rm S}_{\rm BS}}$ denote the analog and digital combining matrices, and ${\bf F}_{\rm RF}^{(t)}\in\mathbb{C}^{N_{\rm UE}\times N^{\rm RF}_{\rm UE}}$ and ${\bf F}_{\rm BB}^{(t)}\in\mathbb{C}^{N^{\rm RF}_{\rm UE}\times N^{\rm S}_{\rm UE}}$ denote the analog and digital precoding matrices, respectively. Here, $N_{\rm UE}^{\rm RF}$ ($N_{\rm BS}^{\rm RF}$) and $N_{\rm UE}^{\rm S}$ ($N_{\rm BS}^{\rm S}$) denote the numbers of RF chains and data streams at the UE (BS), respectively, ${\bf n}^{(t)}\sim\mathcal{CN}(0,\sigma_{\rm n}^2 {\bf I})$ is the AWGN vector, and $\sigma_{\rm n}^2$ is the power of the AWGN at each receive antenna.

\section{Problem Formulation}
\label{sec:formulation}
By denoting ${\bf W}^{(t)} = ({\bf W}_{\rm RF}^{(t)}{\bf W}_{\rm BB}^{(t)})^H$ and ${\bf f}^{(t)} ={\bf F}_{\rm RF}^{(t)}{\bf F}_{\rm BB}^{(t)}{\bf s}^{(t)}$ for notational brevity, the signal model in \eqref{eq:dltrain} can be rewritten as ${\bf y}^{(t)} = \left( ({\bf f}^{(t)})^{T}\otimes {\bf W}^{(t)} \right){\rm vec}({\bf H})+\tilde{\bf n}^{(t)}$, where $\tilde{\bf n}^{(t)} = {\bf W}^{(t)}{\bf n}^{(t)}$. Stacking $\tau$ training slots together, we obtain
\begin{equation}
	{\bf y} = {\boldsymbol{\Phi}{\bf h}}+\tilde{\bf n},
	\label{eq:linearproblem}
\end{equation}
where ${\bf y} = [({\bf y}^{(1)})^{H},\cdots,({\bf y}^{(\tau)})^{H}]^{H}$ is the overall received signal, $\boldsymbol{\Phi} = [(({\bf f}^{(1)})^{T}\otimes {\bf W}^{(1)})^{H},\cdots, (({\bf f}^{(\tau)})^{T}\otimes {\bf W}^{(\tau)})^{H}]^{H}$ is the measurement matrix, and ${\bf h} = {\rm vec}({\bf H})$ is the vectorized downlink channel vector. Estimating $\bf h$ in~\eqref{eq:linearproblem} via linear methods requires excessive training overhead $\tau\geq N_{\rm BS}N_{\rm UE}$, which is infeasible in XL-MIMO systems. In light of this, CS-based reconstruction methods were proposed to fully utilize the intrinsic sparsity of ${\bf H}$~\cite{1614066}, for which the sparse reconstruction problem is formulated as
\begin{equation}
	\begin{aligned}
		{\rm(P1)}\quad\underset{\tilde{\bf h}}{\min}~\Vert \tilde{\bf h} \Vert_0,
		\quad\quad
		{\rm s.t.}~\Vert\boldsymbol{\Phi} \boldsymbol{\Psi} \tilde{\bf h}-\mathbf{y}\Vert_2 \leq \varepsilon,
	\end{aligned}
\end{equation}
where $\tilde{\bf h}$ is the sparse support vector to be estimated, $\varepsilon$ is the maximum tolerable error bound, and ${\boldsymbol{\Psi}}$ is the dictionary matrix. 

\begin{remark}
	Recall that a desirable dictionary $\boldsymbol{\Psi}$ has to satisfy the following properties. 
	First, the dictionary should match the signal model of ${\bf h}$ in~\eqref{eq:linearproblem} to capture inherent features and efficiently sparsify the channel vector as $\tilde{\bf h}$.
	Second, the dictionary should contain as many accurate spatial characteristics as possible, e.g., the locations of the UE and scatterers in the near-field region. In this way, more prior knowledge of the propagation environment included in the dictionary leads to fewer unknown parameters to be estimated, thereby necessitating a shorter pilot length and fewer ensuing baseband samples. Third, the dictionary should be small in size to save storage, thus reducing the computational and memory requirements for real-time processing. 
	To find such a dictionary with minimized size for reducing the implementation complexity, in Sections~\ref{sec:coord_est} and~\ref{sec:ce}, 
	we propose a novel baseband sampling-free localization and an orthogonal dictionary design, respectively.
\end{remark}

\section{Time Inversion Localization Algorithm}
\label{sec:coord_est}
Inspired by holography, which is able to record high-dimensional information with low-dimensional coherent wavefront recordings~\cite{PhysRevLett.118.183901}, we propose to estimate the locations of the UE and scatterers from the measured power pattern in~\eqref{eq:pwrrecord}. In this section, we develop a time inversion algorithm, and prove its feasibility to determine the location coordinates. We further derive that the proposed algorithm can be approximated in a convolution form, which can be implemented by FFT to effectively reduce the computational complexity. Finally, we propose a multipath localization method that iteratively estimates the locations of the UE and scatterers in the near-field region.

\subsection{Time Inversion Near-Field Reconstruction}
In this subsection, we propose a novel method for reconstructing the scattering environment in the near field, focusing on the precise localization of the UE and scatterers. The proposed method is based on the principle that by mathematically emulating the reverse propagation, we can trace the paths of scattered EM waves back to their points of origin, resembling a time inversion, and thereby reconstructing the positions of the UE and scatterers. 

{Since the power of the reference signal {${\bf e}_{\rm ref}$} itself in~\eqref{eq:pwrrecord} contains no information about the UE and scatterers}, we firstly subtract it and define the \textit{hologram} ${\bm \hbar}\in\mathbb{R}^{N_{\rm BS}}$~\cite{Miller} from the power measurement as
\begingroup
\allowdisplaybreaks
\begin{align}
	{\bf p}[n] - \left\vert  {\bf e}_{\rm ref}[n] \right\vert^2 &= {\bf e}_{\rm ref}^{*}[n]{\bf e}_{\rm obj}[n] + {\bf e}_{\rm ref}[n]{\bf e}^{*}_{\rm obj}[n]+ \left\vert {\bf e}_{\rm obj}[n] \right\vert^2\notag \\
	&\gtrsim {\bf e}_{\rm ref}^{*}[n]{\bf e}_{\rm obj}[n] + {\bf e}_{\rm ref}[n]{\bf e}^{*}_{\rm obj}[n]\notag\\
	&\triangleq {\bm \hbar}[n] = 2{\rm Re}\left( {\bf e}_{\rm ref}[n]{\bf e}_{\rm obj}[n] \right),\label{eq:hologram}
\end{align}
\endgroup
where the object power pattern $\left\vert {\bf e}_{\rm obj}[n] \right\vert^2\ll \left\vert  {\bf e}_{\rm ref}[n] \right\vert^2$ is neglected\footnote{Note that the power of the reference wave can be artificially controlled and can be obtained precisely during the calibration procedure. On the other hand, the power of the object wave is inversely proportional to the square of the distance between UE and BS arrays, and is far lower than that of the power of the reference wave.}. Note that the phase information of the object wave was discarded in~\eqref{eq:hologram}. To address this issue, we multiply the hologram ${\bm \hbar}$ by the coherent reference wave as
\begin{equation}
	\overline{{\bm \hbar}}[n] = {\bm \hbar}[n] {\bf e}_{\rm ref}[n] =  A^2 {\bf e}_{\rm obj}[n] + \left( {\bf e}_{\rm ref}[n] \right)^2{\bf e}^{*}_{\rm obj}[n],
	\label{eq:holoref}
\end{equation}
where the phase information of the object wave is restored. 

Then, we can emulate the time inversion EM wave propagation back to any given coordinate\footnote{It is a near-field generalization of the matched filter method, which degenerates into the matched filter method in the far-field region.} ${\bf r}=(x,y)$ by constructing the following function
\begin{equation}
	{ E}({\bf r}) = \sum_{n=1}^{N_{\rm BS}} \overline{{\bm \hbar}}[n] g^{*}\left({\bf r}_{\rm BS}^{(n)},{\bf r}\right)=\sum_{n=1}^{N_{\rm BS}} \overline{{\bm \hbar}}[n]\frac{e^{\jmath\kappa \Vert {\bf r}-{\bf r}_{\rm BS}^{(n)} \Vert}}{\Vert {\bf r}-{\bf r}_{\rm BS}^{(n)} \Vert},
\end{equation}
where the conjugate spatial response $g^{*}(\cdot)$ signifies the propagation in the inverted direction. Note that the denominator of $g(\cdot)$ increases with the distance $\Vert{\bf r}-{\bf r}_{\rm BS}^{(n)}\Vert$, which scales the amplitude of the emulated wave ${ E}({\bf r})$. In this case, the UE or scatterers with larger distances from the BS will be overshadowed due to tiny values of ${ E}({\bf r})$, referred to as the \emph{near-far effect}, causing false estimations especially in the near-field region. We hence compensate the emulated wave as
\begin{equation}
	\tilde{ E}({\bf r}) = \sum_{n=1}^{N_{\rm BS}} \overline{{\bm \hbar}}[n]  g^{*}\left({\bf r}_{\rm BS}^{(n)},{\bf r}\right)\Vert{\bf r}-{\bf r}_{\rm BS}^{(n)}\Vert
	\label{eq:timeinv_compensate}
\end{equation}
to eliminate the near-far effect. 

In other words, we construct the compensated wave via~\eqref{eq:timeinv_compensate} for localization after we receive ${\bf p}[n]$ at the BS. Next, we analyze the properties of the compensated wave $\tilde{E}({\bf r})$. We first rewrite the object wave in~\eqref{eq:obj} as
\begin{equation}
	\begin{aligned}
		&~{\bf e}_{\rm obj}[n]\\
		\approx&~\overline{\alpha}_0g\left( {\overline{\bf r}}_{\rm UE},{\bf r}_{\rm BS}^{(n)}\right) + \sum_{\ell=1}^{L} \overline{\alpha}_\ell {g}\left( {\bf r}_{\rm S}^{(\ell)},{\bf r}_{\rm BS}^{(n)} \right)\triangleq \overline{\bf e}_{\rm obj}[n],
		\label{eq:obj2}
	\end{aligned}
\end{equation}
where ${\overline{\bf r}}_{\rm UE}$ denotes the centroid coordinate\footnote{The aperture of a UE array is typically small in the near-field region. For instance, a $28\,$GHz $8$-antenna ULA UE occupies $3.7\,$cm. Hence, the UE array can be treated as a point during the sensing stage.} of the UE array, $\overline{\alpha}_0$ and $\overline{\alpha}_\ell = \alpha_\ell \sum_{m=1}^{N_{\rm UE}} {g} ( {\bf r}_{\rm S}^{(\ell)},{\bf r}_{\rm UE}^{(m)})$ are the effective channel attenuation for the LoS and NLoS paths, respectively. By defining ${\bf r}_{\rm S}^{(0)} = {\overline{\bf r}}_{\rm UE}$,~\eqref{eq:obj2} can be recast in a compact form as
\begin{equation}
	\overline{\bf e}_{\rm obj}[n] =\sum_{\ell=0}^{L} \overline{\alpha}_\ell  g\left( {\bf r}^{(\ell)}_{\rm S},{\bf r}_{\rm BS}^{(n)} \right).
	\label{eq:obj3}
\end{equation}
Substituting~\eqref{eq:obj3} into~\eqref{eq:timeinv_compensate}, we have
\begingroup
\allowdisplaybreaks
\begin{align}
	\tilde{E}({\bf r}) 
	={} & \sum_{n=1}^{N_{\rm BS}}  \left(\underbrace{\left\vert {\bf e}_{\rm ref}[n] \right\vert^2 \overline{\bf e}_{\rm obj}[n] g^{*}\left({\bf r}_{\rm BS}^{(n)},{\bf r}\right)}_{\rm desired} \right.\notag\\&\left.+ \underbrace{\left( {\bf e}_{\rm ref}[n] \right)^2  \overline{\bf e}_{\rm obj}^{*}[n] g^{*}\left({\bf r}_{\rm BS}^{(n)},{\bf r}\right)}_{\rm interference} \right)\Vert{\bf r}-{\bf r}_{\rm BS}^{(n)}\Vert\notag\\
	={} 
	& A^2 \sum_{\ell=0}^{L} \overline{\alpha}_\ell\left( \sum_{n=1}^{N_{\rm BS}}\frac{ e^{-\jmath\kappa\left( \Vert {\bf r}_{\rm S}^{(\ell)}-{\bf r}_{\rm BS}^{(n)} \Vert- \Vert {\bf r}-{\bf r}_{\rm BS}^{(n)} \Vert\right)}}{ \Vert {\bf r}_{\rm S}^{(\ell)}-{\bf r}_{\rm BS}^{(n)} \Vert } \right.\notag\\
	&+ \left.\frac{  e^{-2\jmath\xi_n -2\jmath \angle \overline{\alpha}_\ell +\jmath\kappa\left( \Vert {\bf r}_{\rm S}^{(\ell)}-{\bf r}_{\rm BS}^{(n)} \Vert+ \Vert {\bf r}-{\bf r}_{\rm BS}^{(n)} \Vert\right)}}{ \Vert {\bf r}_{\rm S}^{(\ell)}-{\bf r}_{\rm BS}^{(n)} \Vert } \right)\notag\\
	& + \sum_{n=1}^{N_{\rm BS}} 2{\rm Re}\left( {\bf e}_{\rm n}[n] e^{\jmath\xi_n} \right)\frac{  e^{\jmath\kappa  \Vert {\bf r}-{\bf r}_{\rm BS}^{(n)} \Vert}}{ \Vert {\bf r}_{\rm S}^{(\ell)}-{\bf r}_{\rm BS}^{(n)} \Vert }
	\notag\\
	={}&A^2 \sum_{\ell=0}^{L} \overline{\alpha}_\ell \left(  \tilde{E}_\ell^{\rm d}({\bf r}) + \tilde{E}_\ell^{\rm i}({\bf r})  \right)+\tilde{E}^{\rm n}({\bf r}),\label{eq:compensated_rec}
\end{align}
\endgroup
where 
\begin{align}
	\tilde{E}_\ell^{\rm d}({\bf r}) &= \sum_{n=1}^{N_{\rm BS}} \frac{ e^{-\jmath\kappa\left( \Vert {\bf r}_{\rm S}^{(\ell)}-{\bf r}_{\rm BS}^{(n)} \Vert- \Vert {\bf r}-{\bf r}_{\rm BS}^{(n)} \Vert\right)}}{ \Vert {\bf r}_{\rm S}^{(\ell)}-{\bf r}_{\rm BS}^{(n)} \Vert }, \\
	\tilde{E}_\ell^{\rm i}({\bf r}) &= \sum_{n=1}^{N_{\rm BS}} \frac{  e^{-2\jmath\xi_n +\jmath\kappa\left( \Vert {\bf r}_{\rm S}^{(\ell)}-{\bf r}_{\rm BS}^{(n)} \Vert+ \Vert {\bf r}-{\bf r}_{\rm BS}^{(n)} \Vert\right)}}{ \Vert {\bf r}_{\rm S}^{(\ell)}-{\bf r}_{\rm BS}^{(n)} \Vert },
\end{align}
and 
\begin{equation}
	\tilde{E}^{\rm n}({\bf r}) = \sum_{n=1}^{N_{\rm BS}} 2{\rm Re}\left( {\bf e}_{\rm n}[n]  \right)\frac{  e^{-\jmath\xi_n +\jmath\kappa  \Vert {\bf r}-{\bf r}_{\rm BS}^{(n)} \Vert}}{ \Vert {\bf r}_{\rm S}^{(\ell)}-{\bf r}_{\rm BS}^{(n)} \Vert }~~~~
\end{equation}
denote the desired term that contributes to the reconstruction, the interference term introduced by the power sampling procedure~\eqref{eq:pwrrecord} of the $\ell$-th path, and the noise term, respectively. 
\begin{prop}
	\label{prop:one_max}
	There exists only one unique location vector ${\bf r}$ that maximizes the power of the desired term $\tilde{E}_\ell^{\rm d}({\bf r})$ as
	\begin{equation}
		{\bf r} = {\bf r}_{\rm S}^{(\ell)}.
	\end{equation}
\end{prop}
\begin{proof}
	For the desired term $\tilde{E}_\ell^{\rm d}({\bf r})$ in~\eqref{eq:compensated_rec}, we can reformulate the power using Cauchy–Schwarz inequality as
	\begin{align}
		\left\vert\tilde{E}^{\rm d}_\ell({\bf r})\right\vert^2&= \left\vert \sum_{n=1}^{N_{\rm BS}} \frac{e^{-\jmath\kappa \Vert {\bf r}_{\rm S}^{(\ell)}-{\bf r}_{\rm BS}^{(n)} \Vert}e^{ \jmath\kappa \Vert {\bf r}-{\bf r}_{\rm BS}^{(n)} \Vert}}{\Vert {\bf r}_{\rm S}^{(\ell)}-{\bf r}_{\rm BS}^{(n)} \Vert} \right\vert^2 \label{eq:rec_pwr}\\
		&\overset{(a)}{\leq} \sum_{n=1}^{N_{\rm BS}}  \frac{\left\vert e^{-\jmath\kappa \Vert {\bf r}_{\rm S}^{(m)}-{\bf r}_{\rm BS}^{(n)} \Vert}\right\vert^2 }{\Vert {\bf r}_{\rm S}^{(\ell)}-{\bf r}_{\rm BS}^{(n)} \Vert} \sum_{n=1}^{N_{\rm BS}} \frac{ \left\vert e^{ \jmath\kappa \Vert {\bf r}-{\bf r}_{\rm BS}^{(n)} \Vert}\right\vert^2}{\Vert {\bf r}_{\rm S}^{(\ell)}-{\bf r}_{\rm BS}^{(n)} \Vert}.\notag
	\end{align}
	The equality of $(a)$ holds if and only if ${\bf r} = {\bf r}_{\rm S}^{(\ell)}$, which indicates that $\vert \tilde{E}^{\rm d}_\ell({\bf r}) \vert$ reaches its global maximum only at position $ {\bf r}_{\rm S}^{(\ell)}$.
\end{proof}
\begin{figure}[t]
	\centering\setcounter{figure}{3}
	\begin{minipage}[t]{0.485\linewidth}
		\centering
		\includegraphics[height=0.165\textheight]{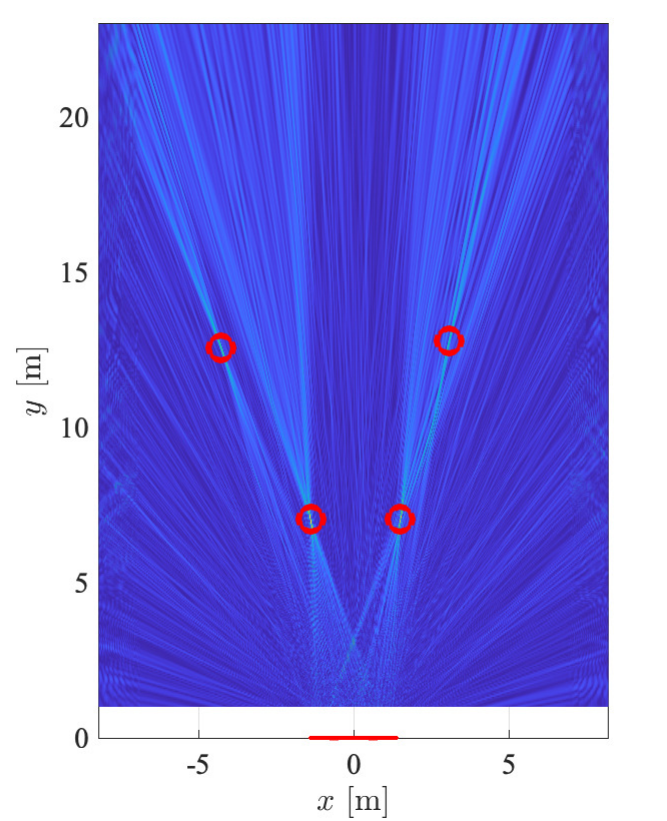}
	\end{minipage}
	\begin{minipage}[t]{0.485\linewidth}
		\centering
		\includegraphics[height=0.165\textheight]{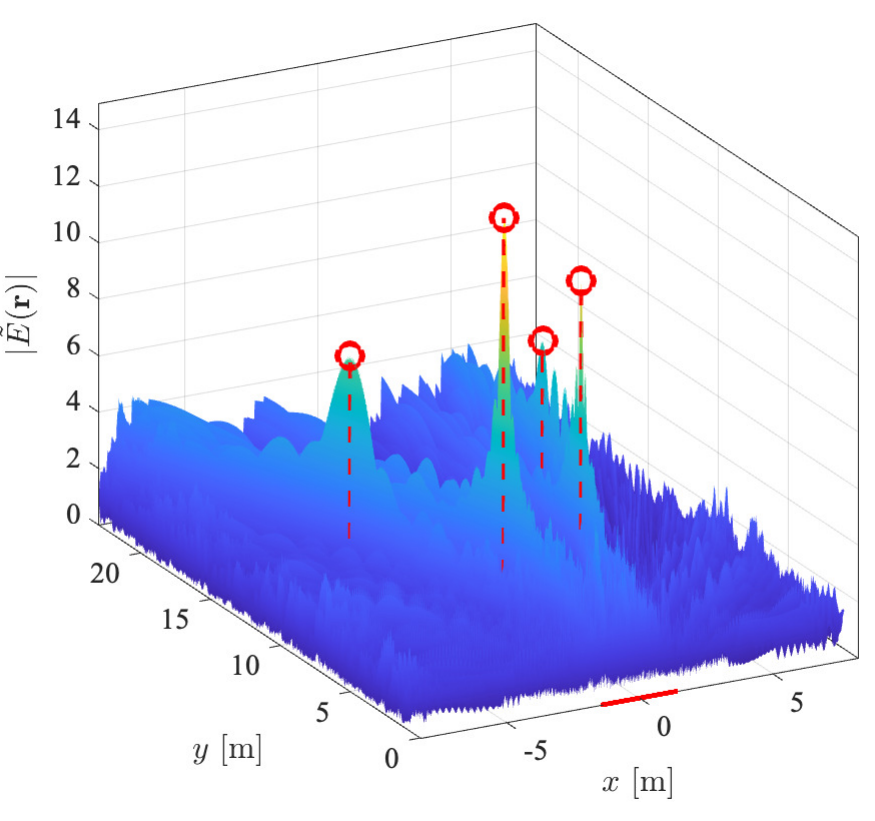}
	\end{minipage}
	\caption{A demonstration of the proposed localization method with $N_{\rm BS}=512$ antenna elements and $L=3$ NLoS paths. The localization is operated with a carrier frequency of $f_c = 28$\:GHz. The red line on the $x$-axis denotes the aperture of XL-MIMO array, while the circles mark the coordinates of the detected multiple sources.}
	\label{fig:rec_example}
\end{figure}
\begin{prop}
	\label{prop:RiemannLebesgue}
	As the wavenumber $\kappa \rightarrow \infty$ and the number of antennas at BS side $N_{\rm BS} \rightarrow \infty$, the desired term vanishes everywhere except the point specified in Proposition~\ref{prop:one_max}:
	\begin{align}
		\lim_{\begin{subarray}{c}\kappa \to {\infty} \\ N_{\rm BS} \to {\infty}\end{subarray}} \!\!\left\vert \tilde{E}_\ell^{\rm d}({\bf r}) \right\vert \!=\! \begin{cases}
			\displaystyle
			\sum_{n=1}^{N_{\rm BS}}\frac{1}{ d^{(\ell,n)} }\!\!\!\!& \text{if}~{\bf r} = {\bf r}_{\rm S}^{(\ell)}\!,\\
			0 & \text{otherwise},
		\end{cases}
		\label{eq:desvanish}
	\end{align}
	where $d^{(\ell,n)} = \Vert {\bf r}_{\rm S}^{(\ell)}-{\bf r}_{\rm BS}^{(n)} \Vert$, and the vanishing effect also holds for both the interference and noise terms as
	\begin{align}
		\lim_{\begin{subarray}{c}\kappa \to {\infty} \\ N_{\rm BS} \to {\infty}\end{subarray}} \left\vert \tilde{E}_\ell^{\rm i}({\bf r}) \right\vert &= 0,
		\label{eq:intfvanish}\\
		\lim_{\begin{subarray}{c} \\ N_{\rm BS} \to {\infty}\end{subarray}} \left\vert \tilde{E}^{\rm n}({\bf r}) \right\vert &= 0,
	\end{align}
	when the reference phase shift $\xi_n$ follows arbitrary distributions satisfying $\mathbb{E}[\cos(2\xi_n)] = 0$ and $\mathbb{E}[\sin(2\xi_n)] = 0$. 
\end{prop}
\begin{proof}
	See proof in Appendix~\ref{sec:appendix2}.
\end{proof}
Note that the distribution constraints in Proposition~\ref{prop:RiemannLebesgue} can be readily satisfied by choosing the distribution of $\xi_n$ as, e.g., a uniform distribution over $[0,\pi]$. Propositions~\ref{prop:one_max} and~\ref{prop:RiemannLebesgue} establish a fundamental understanding of the dynamics in time inversion near-field reconstruction. Specifically, we demonstrate that the desired term will reach its maximum at ${\bf r} = {\bf r}_{\rm S}^{(\ell)}$ for the $\ell$-th path, which is crucial for accurate localization in the near-field region. Besides, both the desired and the interference terms will asymptotically diminish to $0$ everywhere except the points specified in Proposition~\ref{prop:one_max} for the desired part. A demonstration of near-field localization by the proposed time inversion algorithm is shown in Fig.~\ref{fig:rec_example}, where the location coordinates of the UE and $L=3$ scatterers are accurately estimated. Next, we will leverage these properties to develop a low-complexity time inversion algorithm for localization in near-field XL-MIMO.

{
\begin{remark}
	In XL-MIMO systems, spatial non-stationarity can occur, meaning that some multipath components are detected only by certain antennas, but not available for other antennas. 
	Our proposed method is also applicable to these potential spatial non-stationary scenario, although the reduced effective aperture of the visible array may lead to higher localization errors. With the widely adopted orthogonal frequency division multiplexing (OFDM) waveform, we can aggregate the reconstruction results over multiple subcarriers, which effectively improves the distance resolution and suppresses the interference to relieve the performance drain caused by spatial non-stationarity~\cite{GLOBECOM}.
\end{remark}}
\subsection{Low-Complexity Time Inversion Algorithm}
The locations of the UE and scatterers can be reconstructed exploiting the proposed time inversion method by enumerating possible target locations ${\bf r}$ in the near-field grid. Specifically, during the sensing stage, the BS array generates a uniformly distributed vector $\boldsymbol{\xi}=[\xi_1,\cdots,\xi_{N_{\rm BS}}]$ and performs phase shift $e^{\jmath\xi_n}$ at the $n$-th antenna element. The near-field region is split into $G_x\times G_y$ grids, where $G_x$ ($G_x>N_{\rm BS}$) and $G_y$ denote the numbers of grids on the $x$-axis and the $y$-axis, respectively. The compensated wave at the $i$-th grid on $x$-axis and the $j$-th grid on $y$-axis, i.e., ${\bf r}^{(i,j)} = (x_i,y_j)$, is finally manipulated by the time inversion method using~\eqref{eq:timeinv_compensate}, and is then stored into the $i$-th row and the $j$-th column of the amplitude matrix $\boldsymbol{\Theta}\in\mathbb{C}^{G_x\times G_y}$, where
\begin{equation}
	\boldsymbol{\Theta}[i,j] = \tilde{E}\left({\bf r}^{(i,j)}\right).
	\label{eq:theta}
\end{equation}
The coordinates of the UE and scatterers can then be estimated as
\begin{equation}
	\left(\hat{x}^{(\ell)},\hat{y}^{(\ell)}\right)\!=\!\left( x_L+{i}^*\frac{x_H-x_L}{G_x}, y_L+j^*\frac{y_H-y_L}{G_y} \right),
	\label{eq:location2}
\end{equation}
where
\begin{equation}
	\begin{aligned}
		({i^*},{j^*}) = {\rm arg}\underset{i,j}{\rm max}\ \left\vert \boldsymbol{\Theta}[i,j] \right\vert^2.
	\end{aligned}
	\label{eq:location}
\end{equation}
However, constructing $\boldsymbol{ \Theta }$ in~\eqref{eq:theta} by enumerating all possible near-field coordinates ${\bf r}^{(i,j)} = (x_{i},y_{j})^T$ on the $G_x\times G_y$ grids can be time-consuming, especially when the number of grids and $N_{\rm BS}$ are extremely large. To effectively reduce the number of computations in the proposed time inversion algorithm, we expand the procedure~\eqref{eq:timeinv_compensate} via a grid-search manner as
\begin{equation}
	\begin{aligned}
		\tilde{E}\left({\bf r}^{(i,j)}\right) = \sum_{n=1}^{N_{\rm BS}} {{\bm \hbar}}[n] {\bf e}_{\rm ref}[n] e^{\jmath\kappa \sqrt{\left( x_{\rm BS}^{(n)} - x_i \right)^2+y_j^2}},
	\end{aligned}
\end{equation}
which can be reformulated as the linear convolution of vector $\overline{{\bm \hbar}} = {{\bm \hbar}}\odot {\bf e}_{\rm ref} \in\mathbb{C}^{N_{\rm BS}}$ defined in~\eqref{eq:holoref} and $\overline{\bf g}_{j}= [e^{\jmath\kappa \sqrt{x_1^2+y_j^2}} , \cdots , e^{\jmath\kappa \sqrt{x_{G_x}^2+y_j^2}}]^T\in\mathbb{C}^{G_x\times 1}$. Hence, for a fixed $j$-th grid on $y$-axis, the time inversion reconstructed result\footnote{Convolution extends the dimension to $N_{\rm BS}+G_x-1$, which leads to incomplete reconstruction at the edges, and the valid results are obtained only for indices from $N_{\rm BS}$ to $G_x$. To satisfy the need for $G_x$ valid points, the $x$-axis grid will be extended to $G_x^\prime = G_x + 2N_{\rm BS}$ to compensate the associated edge effects.} can be recast as
\begin{equation}
	\boldsymbol{\Theta}[:,j] = \left\vert \left( \overline{{\bm \hbar}}  \circledast \overline{\bf g}_j \right)_{N_{\rm BS}:G_x^\prime-N_{\rm BS}-1} \right\vert,
	\label{eq:convrec}
\end{equation}
where $\overline{\bf g}_{j}= [e^{\jmath\kappa \sqrt{x_1^2+y_j^2}},\cdots,e^{\jmath\kappa \sqrt{x_{G_x^\prime}^2+y_j^2}}]^T\in\mathbb{C}^{G_x^\prime\times 1}$ with $G_x^\prime = G_x+2N_{\rm BS}$. Notably,~\eqref{eq:convrec} can be further implemented efficiently using FFT as
\begin{equation}
	\boldsymbol{\Theta}[:,j] = \left\vert \mathcal{F}^{-1}_{G_x^\prime} \left[\mathcal{F}_{G_x^\prime} \left[ \overline{{\bm \hbar}}  \right]  \odot \mathcal{F}_{G_x^\prime} \left[ \overline{\bf g}_j  \right]  \right]_{N_{\rm BS}:G_x^\prime-N_{\rm BS}-1} \right\vert,
	\label{eq:fftrec}
\end{equation}
where $\mathcal{F}_{a}[\cdot]$ denotes an FFT operation with padding length $2^{\lceil \log_2 a\rceil}$, and $\mathcal{F}^{-1}_{a}[\cdot]$ denotes an inverse FFT (IFFT) operation with trimming length $a$. By resorting to FFT, the computational complexity can be greatly reduced from $\mathcal{O}(N_{\rm BS} G_x)$ to $\mathcal{O}(P\log P)$, where $P = 2^{\lceil \log_2 G_x^\prime \rceil}$ is the nearest power-$2$ number larger than $G_x^\prime$.
\begin{algorithm}[t]
	\caption{Time Inversion Multipath Localization}\label{alg:TIrec}
	\begin{algorithmic}[1]
		\REQUIRE The range of near-field region $x_L$, $x_H$, $y_L$, and $y_H$, the numbers of grids $G_x$ and $G_y$, multipaths $L$,\footnotemark~and antennas $N_{\rm BS}$. Measured power pattern $\bf p$ and reference wave ${\bf e}_{\rm ref}$.
		\ENSURE The estimated coordinates $\{\hat{\bf r}_{\rm S}^{(\ell)} = (\hat{x}^{(\ell)},\hat{y}^{(\ell)}) \}_{\ell = 0}^{L}$ of the UE and scatterers.
		\STATE Draw uniformly distributed random phase ${\boldsymbol{\xi}}\sim \mathcal{U}[0,2\pi]$.
		\STATE Initialize the hologram ${{\bm \hbar}}_0[n] = {\bf p}[n] - \left\vert  {\bf e}_{\rm ref}[n] \right\vert^2$.
		\FOR{$\ell = 0,\cdots,L$}
		\STATE Initialize the empty amplitude matrix ${\boldsymbol{\Theta}}\in\mathbb{R}^{G_x\times G_y}$.
		\FOR{$j = 1,\cdots,G_y$}
		\STATE Fill the $j$-th column according to~\eqref{eq:fftrec}.
		\ENDFOR
		\STATE Find the maximal grid as
		\vspace{-2mm}
		\begin{equation}
			({i^*},{j^*}) = {\rm arg}\underset{i,j}{\rm max}\ \left\vert \boldsymbol{\Theta}[i,j] \right\vert^2.\notag
		\end{equation}
		\vspace{-5mm}
		\STATE Estimate the $\ell$-th coordinate as
		\begin{equation}
			\left(\hat{x}^{(\ell)},\hat{y}^{(\ell)}\right)\!=\!\left( x_L\!+\!{i}^*\frac{x_H-x_L}{G_x}, y_L\!+\!j^*\frac{y_H-y_L}{G_y} \right).\notag
		\end{equation}
		\vspace{-3mm}
		\STATE Update the hologram according to~\eqref{eq:updateholo}.
		\ENDFOR
		\STATE Reture the coordinates of UE and scatterers $\{\hat{\bf r}_{\rm S}^{(\ell)} =  (\hat{x}^{(\ell)},\hat{y}^{(\ell)}) \}_{\ell = 0}^{L}$.
	\end{algorithmic}
\end{algorithm}
\footnotetext{
	To obtain the value of $L$, an empirical threshold can be adopted to identify the peaks in $\boldsymbol{\Theta}$ that are significantly higher than other ripples. 
}

\subsection{Proposed Multipath Localization}

According to~\eqref{eq:compensated_rec} and~\eqref{eq:desvanish}, the reconstructed channel attenuation at the coordinate of the $\ell$-th scatterer is given by
\begin{equation}
	\hat{\alpha}_{\ell}  = \overline{\alpha}_\ell  \sum_{n=1}^{N_{\rm BS}} \frac{1}{d^{(\ell,n)} }.
\end{equation}
By defining ${{\bm \hbar}}_0[n] = {\bf e}_{\rm ref}^{*}[n]{\bf e}_{\rm obj}[n] + {\bf e}_{\rm ref}[n]{\bf e}^{*}_{\rm obj}[n]$ and assuming that the power of multipath components $\hat{\alpha}_{\ell} $ decreases with $\ell$ as $\hat{\alpha}_{0} \geq \hat{\alpha}_{1} \geq\cdots\geq \hat{\alpha}_{L} $, we can detect the location coordinates from the $\ell$-th path by iteratively deleting the multipath components as
\begin{equation}
	\begin{aligned}
		{{\bm \hbar}}_{\ell+1}[n] ={}& {{\bm \hbar}}_{\ell}[n] - \left( {\bf e}_{\rm ref}^{*}[n]\hat{\alpha}_{\ell}  g\left( {\bf r}^{(\ell)}_{\rm S},{\bf r}_{\rm BS}^{(n)} \right) \right.\\&+{\bf e}_{\rm ref}[n]\left.\hat{\alpha}_{\ell}^* g^*\left( {\bf r}^{(\ell)}_{\rm S},{\bf r}_{\rm BS}^{(n)} \right) \right)\\
		={}&{\bf e}_{\rm ref}^*[n] \sum_{\ell^\prime=\ell+1}^{L} \hat{\alpha}_{\ell^\prime}  g\left( {\bf r}^{(\ell^\prime)}_{\rm S},{\bf r}_{\rm BS}^{(n)} \right) \\&+{\bf e}_{\rm ref}[n] \sum_{\ell^\prime=\ell+1}^{L} \hat{\alpha}_{\ell^\prime}^* g^*\left( {\bf r}^{(\ell^\prime)}_{\rm S},{\bf r}_{\rm BS}^{(n)} \right).
		\label{eq:updateholo}
	\end{aligned}
\end{equation}
With the $\ell$-th hologram pattern ${{\bm \hbar}}_\ell$ containing $L-\ell +1$ multipath components, we iteratively apply the proposed time inversion algorithm to detect the location coordinates of the scatterers in a power decreasing manner. The overall localization procedure is summarized in Algorithm~\ref{alg:TIrec}.

\section{Channel Estimation with Locations}
\label{sec:ce}
In this section, we propose a novel dictionary design algorithm for CE that fully utilizes the estimated locations of the UE and scatterers in Section~\ref{sec:coord_est}. 
Besides, the detected location coordinates of the scatterers are also exploited in estimating the channel attenuations of NLoS paths.

\subsection{Eigen-Dictionary Design for Sparse Representation}
For near-field wireless communications, it is highly likely that the LoS component dominates the channel gain~\cite{9139337}. Note that the UE location detected by the time inversion algorithm only provides a coarse estimation, i.e., $\hat{\bf r}_{\rm S}^{(0)}$, of the coordinates ${\bf r}_{\rm UE}^{(m)}$ in the LoS component of the near-field channel model~\eqref{eq:LoS}. In this subsection, we develop a fine-grained CE for the LoS path by proposing a novel dictionary design to combat the model mismatch between far-field and near-field channel models. Specifically, we perform EVD on the auto-correlation matrix of the near-field channel and derive the general form of the eigenvectors. The proposed eigen-dictionary is therefore constructed based on the estimated location of the UE and eigenvectors.

Recall that problem $\rm (P1)$ requires the identification of a dictionary $\boldsymbol{\Psi}$ that efficiently sparsifies the near-field channel matrix. In this regard, the singular value decomposition (SVD) decomposes the channel matrix in the form of ${\bf H} = {\bf U}{\boldsymbol{\Sigma}}{\bf V}^{H}$, where ${\bf H}$ can be properly sparsified to a diagonal singular value matrix ${\boldsymbol{\Sigma}}$ by unitary matrices ${\bf U}$ and ${\bf V}$. More importantly, the resulting dictionary $\boldsymbol{\Psi} = {\bf V}^*\otimes {\bf U}$ also shows mutual orthogonality between codewords. Inspired by the SVD, we consider designing the dictionary matrix exploiting the singular vectors. Since the channel matrix is not a square matrix when $N_{\rm UE}\neq N_{\rm BS}$, singular vectors can be obtained separately from the corresponding EVD of the auto-correlation matrices. According to~\eqref{eq:channelmodel}, the auto-correlation matrix at the BS side is given by
\begin{equation}
	\begin{aligned}
		{\bf R}_{\rm BS} & \overset{~~~}{=} \mathbb{E}\left[{\bf H} {\bf H}^{H}\right]\\
		& \overset{~~~}{=} {\bf H}_{\rm LoS}{\bf H}_{\rm LoS}^{H}+\mathbb{E}\left[ {\bf H}_{\rm NLoS}{\bf H}_{\rm NLoS}^{H} \right]\\
		&\overset{~~~}{=} {\bf H}_{\rm LoS}{\bf H}_{\rm LoS}^{H}+
		\boldsymbol{ \Gamma },
	\end{aligned}
	\label{eq:autocorr}
\end{equation}
where $\boldsymbol{ \Gamma } = {\rm diag}(\gamma_1,\cdots,\gamma_{N_{\rm BS}})$ is a real-valued diagonal matrix with 
\begin{equation}
	\gamma_n = \frac{1}{L} \sum_{\ell=1}^{L}\sum_{m}^{N_{\rm UE}}\frac{1}{\Vert {\bf r}_{\rm UE}^{(m)}-{\bf r}_{\rm S}^{(\ell)} \Vert^2 \Vert {\bf r}_{\rm BS}^{(n)}-{\bf r}_{\rm S}^{(\ell)} \Vert^2}.
\end{equation}
Note that $\boldsymbol{ \Gamma }$ in~\eqref{eq:autocorr} is a diagonal matrix and has no impact on calculating eigenvectors since it only adds $\gamma_n$ to the $n$-th eigenvalue. Therefore, the element located at the $n^\prime$-th row and $n$-th column of ${\bf R}_{\rm BS}$ can be expressed by
\begin{equation}
	\begin{aligned}
		&{\bf R}_{\rm BS}\left[n^\prime,n\right] \\
		={}&  {\mathbf g}_{\rm UE}^H\left({\bf r}_{\rm BS}^{(n)}\right){\mathbf g}_{\rm UE}\left({\bf r}_{\rm BS}^{(n)}\right)+\gamma_n\mathds{1}\{n=n^\prime\}\\={}&\gamma_n\mathds{1}\{n=n^\prime\}\!+\! \sum_{m=1}^{N_{\rm UE}}\frac{e^{-\jmath\kappa \Vert{\bf r}_{\rm BS}^{(n)}-{\bf r}_{\rm UE}^{(m)}\Vert }}{\Vert{\bf r}_{\rm BS}^{(n)}-{\bf r}_{\rm UE}^{(m)}\Vert}  \frac{e^{\jmath\kappa \Vert{\bf r}_{\rm BS}^{(n^\prime)}-{\bf r}_{\rm UE}^{(m)}\Vert}}{\Vert{\bf r}_{\rm BS}^{(n^\prime)}-{\bf r}_{\rm UE}^{(m)}\Vert},
	\end{aligned}
	\label{eq:autocorr2}
\end{equation}
where $\mathds{1}\{\cdot\}$ denotes the indicator function. Introducing the near-field paraxial approximation\cite{Miller}, we have
\begingroup
\allowdisplaybreaks
\begin{align}
	&~{\bf R}_{\rm BS}[n^\prime,n]\notag\\ \approx 
	&~\gamma_n\mathds{1}\{n=n^\prime\} +\frac{1 }{r_0^2}\sum_{m=1}^{N_{\rm UE}}e^{-\jmath\kappa\frac{\left(x_{\rm BS}^{(n)}-x_{\rm UE}^{(m)}\right)^2-\left(x_{\rm BS}^{(n^\prime)}-x_{\rm UE}^{(m)}\right)^2}{2 \hat{y}^{(0)} }}\notag\\
	=&~\gamma_n\mathds{1}\{n=n^\prime\}+\frac{  e^{\jmath\kappa \frac{\left(x_{\rm BS}^{(n^\prime)}\right)^2-\left(x_{\rm BS}^{(n)}\right)^2}{2 \hat{y}^{(0)} }}}{r_0^2 } \sum_{m=1}^{N_{\rm UE}} e^{\jmath\kappa\frac{x_{\rm UE}^{(m)}\left(x_{\rm BS}^{(n)}-x_{\rm BS}^{(n^\prime)}\right)}{\hat{y}^{(0)} } }\notag\\
	\triangleq&~\gamma_n\mathds{1}\{n=n^\prime\}+  e^{\jmath\kappa \frac{\left(x_{\rm BS}^{(n^\prime)}\right)^2-\left(x_{\rm BS}^{(n)}\right)^2}{2 \hat{y}^{(0)} }} \overline{\bf R}_{\rm BS}\left[n^\prime,n\right],\label{eq:parax}
\end{align}
\endgroup
where $r_0=\| \hat{\bf r}_{\rm S}^{(0)} \|$ is an approximation of the distance term in the denominator of the second term in \eqref{eq:autocorr2}. Furthermore, $\hat{y}^{(0)}$ denotes the estimated $y$-coordinate of the center of the UE array, which is obtained from Algorithm~\ref{alg:TIrec}. 
Hence, we can rewrite the EVD procedure of ${\bf R}_{\rm BS}$ as 
\begin{equation}
	{\bf R}_{\rm BS}{\bf u}_n = {\bf D}_{\rm BS}^{-1} \left(  \overline{\mathbf R}_{\rm BS}+{\boldsymbol{\Gamma}} \right){\bf D}_{\rm BS}{\bf u}_n = \lambda_n{\bf u}_n,
	\label{eq:evd}
\end{equation}
where ${\bf u}_n$ is the $n$-th eigenvector of ${\bf R}_{\rm BS}$ and $\lambda_n\geq 0$ is the corresponding eigenvalue. According to~\eqref{eq:parax}, we define a compensation matrix ${\bf D}_{\rm BS}$ given by
\begin{equation}
	{\bf D}_{\rm BS} = {\rm diag}\left(e^{\jmath\kappa \frac{\left(x_{\rm BS}^{(1)}\right)^2}{2 \hat{y}^{(0)} }},\cdots,e^{\jmath\kappa \frac{\left(x_{\rm BS}^{(N_{\rm BS})}\right)^2}{2 \hat{y}^{(0)} }}\right),
	\label{eq:compensation}
\end{equation}
which can be well determined according to the BS antenna array configuration.
\begin{algorithm}[t]
	\caption{Proposed Dictionary Design Algorithm}\label{alg:a1}
	\begin{algorithmic}[1]
		\REQUIRE Estimated coordinate $(\hat{x}^{(0)},\hat{y}^{(0)})$ and the numbers of antennas $N_{\rm BS}$ and $N_{\rm UE}$.
		\ENSURE The DPSS-based eigen-dictionary $\boldsymbol{\Psi}_E$.
		\STATE Estimate the compensation matrix ${\bf D}_{\rm BS}$ and ${\bf D}_{\rm UE}$ according to~\eqref{eq:compensation}.
		\STATE Calculate the frequency ${\hat W}=\kappa L_{\rm R}/(4\pi \hat{y}^{(0)})$.
		\STATE Generate $\overline{\bf R}_{\rm BS}$ and $\overline{\bf R}_{\rm UE}$ with DPSS according to \eqref{eq:toeplitzmat}.
		\STATE Perform EVD on $\overline{\bf R}_{\rm BS}=\mathbf{U}\boldsymbol{\Lambda}\mathbf{U}^{-1}$ and $\overline{\bf R}_{\rm UE}=\mathbf{V}{\boldsymbol{\Lambda}^\prime}\mathbf{V}^{-1}$.
		\STATE Compensate the phase shift ${\bf U}^c={\bf D}_{\rm BS}^{-1} {\bf U}$, ${\bf V}^c\!=\!{\bf D}_{\rm UE}^{-1} {\bf V}$.
		\STATE Return eigen-dictionary ${\boldsymbol{\Psi}}_E = ({\bf V}^c)^* \otimes {\bf U}^c$.
	\end{algorithmic}
\end{algorithm}
Note that extracting ${\bf D}_{\rm BS}$ from ${\bf R}_{\rm BS}$ only alters the phase of each eigenvector since $(\overline{\mathbf R}_{\rm BS}+ {\boldsymbol{\Gamma}} ){\bf D}_{\rm BS}{\bf u}_n = \lambda_n {\bf D}_{\rm BS} {\bf u}_n$. Furthermore, \eqref{eq:parax} yields
\begin{equation}
	\begin{aligned}
		\overline{\bf R}_{\rm BS}[n^\prime,n] & \overset{~~~}{=} \frac{1}{r_0^2} \sum_{m=1}^{N_{\rm UE}} e^{\jmath\kappa\frac{x_{\rm UE}^{(m)}\left(x_{\rm BS}^{(n)}-x_{\rm BS}^{(n^\prime)}\right)}{\hat{y}^{(0)} } }\\
		&\overset{(b)}{\simeq} \frac{1}{r_0^2} \int_{-\frac{L_{\rm UE}}{2}}^{\frac{L_{\rm UE}}{2}}e^{\frac{\jmath\kappa}{\hat{y}^{(0)} }x (x_{\rm BS}^{(n)}-x_{\rm BS}^{(n^\prime)})} {\rm d}x \\
		&\overset{~~~}{=}\frac{ 2\hat{y}^{(0)} \sin\left[ \frac{\kappa L_{\rm UE} (x_{\rm BS}^{(n)}-x_{\rm BS}^{(n^\prime)} )}{2 \hat{y}^{(0)}  } \right]  }{ r_0^2 \kappa \left(x_{\rm BS}^{(n)}-x_{\rm BS}^{(n^\prime)} \right)}\\ &\overset{~~~}{\propto}\frac{\sin\left[2\pi W (x_{\rm BS}^{(n)}-x_{\rm BS}^{(n^\prime)} )\right]  }{ \left(x_{\rm BS}^{(n)}-x_{\rm BS}^{(n^\prime)} \right)},
	\end{aligned}
	\label{eq:toeplitzmat}
\end{equation}
where $L_{\rm UE}=(N_{\rm UE}-1)\lambda/2$ denotes the aperture of the UE array with half-wavelength antenna spacing. Note that $(b)$ holds asymptotically when $N_{\rm UE}$ is sufficiently large. Also, $\overline{\bf R}_{\rm BS}$ is a Toeplitz matrix with each column composed of a \textit{shifted sinc function}, and the $n$-th eigenvector ${\bf D}_{\rm BS} {\bf u}_n$ of this matrix is called the $(n-1)$-th order \emph{discrete prolate spheroidal sequence} within frequency $W=\kappa L_{\rm UE}/(4\pi \hat{y}^{(0)})$~\cite{6373401}. 
\begin{remark}
	Recall that a satisfactory dictionary can be constructed by eigenvectors of the auto-correlation matrix ${\bf R}_{\rm BS}$. However, numerically estimating the auto-correlation matrix typically requires a large number of samples. Fortunately, with the analytical result in~\eqref{eq:toeplitzmat}, the auto-correlation matrix can be well-determined directly by a series of sinc functions given the frequency $W$, from which we can readily generate the dictionary by applying an efficient EVD operation. 
\end{remark}
Similarly, we can calculate the eigenvectors $\{ {\bf u}_m \}_{m=1}^{N_{\rm UE}}$ of ${\bf R}_{\rm UE} = \mathbb{E}[{\bf H}^{H}{\bf H}]\approx\mathbf{V}\boldsymbol{\Lambda}^\prime\mathbf{V}^{H}\in\mathbb{C}^{N_{\rm UE}\times N_{\rm UE}}$ and finally form the eigen-dictionary matrix as ${\boldsymbol{\Psi}}_E = (\mathbf{D_{\rm UE}V}^*) \otimes \mathbf{D}_{\rm BS}^{-1}{\bf U} \in\mathbb{C}^{N_{\rm UE}N_{\rm BS}\times N_{\rm UE}N_{\rm BS}}$ for problem $\rm (P1)$. The overall dictionary design procedure is summarized in Algorithm~\ref{alg:a1}. By resorting to the EVD tailored for the near-field channel matrix, we can effectively eliminate the mismatch issue associated with the DFT dictionary. Moreover, the proposed DPSS-based eigen-dictionary naturally holds orthogonality among columns, since both $\bf V$ and $\bf U$ are unitary matrices. This is one of the key advantages compared to the spherical dictionary, which shall be validated via simulation in the next section. 

\subsection{Proposed Multipath Near-Field Channel Estimation}
\label{sec:proposed}
When calculating the auto-correlation matrix in~\eqref{eq:autocorr}, the impacts caused by NLoS paths are averaged out, which indicates that the eigen-dictionary $\boldsymbol{ \Psi}_E$ is only applicable for the LoS path. Therefore, in this subsection we propose an efficient CE scheme, where the NLoS paths are also taken into consideration.

The NLoS channel model in~\eqref{eq:NLoS} can be reformulated in a compact matrix form as
\begin{equation}
	\begin{aligned}
		{\bf H}_{\rm NLoS} = {\bf G}_{\rm UE}\overline{\boldsymbol{\Lambda}}{\bf G}_{\rm BS}^{T},
		\label{eq:NLoS2}
	\end{aligned}
\end{equation}
where ${\bf G}_{\rm UE} = [{\bf g}_{\rm UE}({\bf r}_{\rm S}^{(1)}),\cdots,{\bf g}_{\rm UE}({\bf r}_{\rm S}^{(L)})]\in\mathbb{C}^{N_{\rm UE}\times L}$ and ${\bf G}_{\rm BS} = [{\bf g}_{\rm BS}({\bf r}_{\rm S}^{(1)}),\cdots,{\bf g}_{\rm BS}({\bf r}_{\rm S}^{(L)})]\in\mathbb{C}^{N_{\rm BS}\times L}$ denote the near-field channel steering matrices at the UE side and BS side, respectively, and $\overline{\boldsymbol{\Lambda}} = {\rm diag}(\boldsymbol{\alpha})= {\rm diag}(\alpha_1,\cdots,\alpha_L)$ is the channel attenuation matrix. Vectorizing the NLoS channel matrix in~\eqref{eq:NLoS2} yields
\begin{equation}
	\begin{aligned}
		{\rm vec}\left({\bf H}_{\rm NLoS}\right) = \left( {\bf G}_{\rm BS}\otimes {\bf G}_{\rm UE} \right) {\rm vec}(\overline{\boldsymbol{\Lambda}})={\boldsymbol{\Psi}}_G \boldsymbol{\alpha},
		\label{eq:NLoS3}
	\end{aligned}
\end{equation}
where the overall steering matrix 
\begin{equation}
	\begin{aligned}
		{\boldsymbol{ \Psi }}_G=&\left[
			{\bf g}_{\rm BS}\left({\bf r}_{\rm S}^{(1)}\right)\otimes {\bf g}_{\rm UE}\left({\bf r}_{\rm S}^{(1)}\right) , \cdots ,\right.\\
			&~\left.{\bf g}_{\rm BS}\left({\bf r}_{\rm S}^{(L)}\right) \otimes {\bf g}_{\rm UE}\left({\bf r}_{\rm S}^{(L)}\right)
		\right]\in\mathbb{C}^{N_{\rm BS}N_{\rm UE}\times L}
	\end{aligned}
\end{equation}
denotes the remaining columns of ${\bf G}_{\rm BS}\otimes {\bf G}_{\rm UE}$ under the column selection of ${\rm vec}(\overline{\boldsymbol{\Lambda}})$, and can be readily obtained by the estimated coordinates $\{\hat{\bf r}_{\rm S}^{(\ell)}\}_{\ell=1}^L$.

Substituting~\eqref{eq:channelmodel} and~\eqref{eq:NLoS3} into~\eqref{eq:linearproblem}, the measurement model can then be reformulated as
\begin{equation}
	\begin{aligned}
		{\bf y} &= \boldsymbol{\Phi}\left(  \boldsymbol{\Psi}_G \boldsymbol{\alpha}  +  \boldsymbol{\Psi}_E  {\rm vec}(\overline{\boldsymbol{\Sigma}})  \right)+\tilde{\bf n}\\
		&=\boldsymbol{\Phi}\left( \boldsymbol{\Psi}_{G,E} \left[\boldsymbol{\alpha}^T,  {\rm vec}\left(\overline{\boldsymbol{\Sigma}}\right)^T\right]^T  \right)+\tilde{\bf n},
	\end{aligned}
\end{equation}
where $\overline{\boldsymbol{\Sigma}} = {\bf D}_{\rm UE}\boldsymbol{\Sigma}{\bf D}_{\rm BS}^T$, and ${\boldsymbol{\Psi}}_{G,E} = [\boldsymbol{\Psi}_{G},\boldsymbol{\Psi}_{E}]\in\mathbb{C}^{N_{\rm BS}N_{\rm UE}\times (N_{\rm BS}N_{\rm UE}+L)}$. In other words, according to \eqref{eq:linearproblem}, the matrix ${\boldsymbol{\Psi}}_{G,E}$ can effectively sparsify the near-field channel and therefore can serve as the overall dictionary for CE.

\section{Simulation Results}
In this section, we evaluate the localization and CE performance based on the proposed methods via numerical simulations. The localization performance is evaluated by root mean square error (RMSE) as ${\rm RMSE} = \sqrt{\mathbb{E}[ \Vert \hat{\bf r}-{\bf r} \Vert^2 ] }$, 
where $\hat{\bf r}$ is an estimation of ${\bf r}$. The CE performance is evaluated by normalized mean square error (NMSE) as ${\rm NMSE}( \hat{\bf H} ,{\bf H} ) = \mathbb{E}[{\Vert \hat{\bf H} -{\bf H} \Vert_F^2}/{\Vert {\bf H} \Vert_F^2}]$, 
where $\Vert\cdot\Vert_F$ is the Frobenius norm, and $\hat{\bf H}$ is an estimation of ${\bf H}$.
\begin{table*}[ht]
	\centering
	\caption{Details of Different Localization Algorithms}
	\label{tab:loc}
	\begin{threeparttable}
		\begin{tabular}{c|c|c|ccc}
			\hline\hline
			\multirow{2}{*}{} & \multirow{2}{*}{Baseband Samples} & \multirow{2}{*}{Computational Complexity (Number of Multiplications)} & \multicolumn{3}{c}{Average Runtime (s)} \\ \cline{4-6} 
			&                       &                          & $\eta = 1$      & $\eta = 2$     & $\eta = 3$     \\ \hline
			MUSIC                 & $N_{\rm BS}/2$                      &  $N_{\rm BS}^3+\left(N_{\rm BS}-L+N_{\rm UE} +G_xG_y(N_{\rm BS}-L+1) \right)N_{\rm BS}^2$                        &   $0.8572$                          &  $3.5272$                            &  $8.0858$     \\ \hline
			{BF}                 &  $N_{\rm BS}/2$                &  $N_{\rm BS}G_x G_y+N_{\rm BS}L(L+1)(2L+1)/6+L^2(L+1)^2/4$    &   $0.0548$                           &  $0.1984$                            &   $0.4232$    \\ \hline
			{STT*}                 &  $N_{\rm BS}N_{\rm UE} T_s$                &  $8 T_a\left( N_{\rm BS}+N_{\rm UE} \right)+2N_{\rm BS}G_xG_y$    &   $-$                           &  $-$                            &   $-$    \\ \hline
			Proposed      & $0$           &  $G_y\lceil \log_2 G_x^\prime \rceil 2^{\lceil \log_2 G_x^\prime \rceil}$      & $\bf 0.0022$              &  $\bf 0.0042$                 &  $\bf 0.0065$     \\ \hline\hline
		\end{tabular}
		\begin{tablenotes}
			\footnotesize
			\item[*] $T_a$ denotes the number of training iterations.
		\end{tablenotes}
	\end{threeparttable}
	\vspace{-3mm}
\end{table*}
\subsection{Simulation Setup}
\label{sec:setup}
Throughout the simulation, we consider the arrays at the BS and UE are equipped with $N_{\rm BS}=256$ and $N_{\rm UE} = 4$ antennas with half-wavelength spacing\footnote{Antenna spacing no greater than half a wavelength effectively avoids $2\pi$ phase wrapping effect.}, respectively, and the carrier frequency is set as $f_c = 28\,{\rm GHz}$. The BS array is placed symmetrically on the $x$-axis and the UE is in the near-field region of the BS as shown in Fig.~\ref{fig:sysmodela}. Unless otherwise specified, we deploy a single RF chain at both the BS and UE. The distances from the UE and $L=5$ scatterers to the center of the BS array are selected uniformly from a rectangular area limited by $x_L=-5\,{\rm m}$, $x_H=5\,{\rm m}$, $y_L = 2\,{\rm m}$, and $y_H=25\,{\rm m}$ within the near-field region. {The thermal noise power is set as $\sigma_{\rm n}^2 = -101\,{\rm dBm}$.}

\subsection{Power Sensor-based Localization}

{In this subsection, we compare the performance of the proposed time inversion algorithm with the brute force (BF) grid search, MUSIC algorithm~\cite{10149471}, and STT method~\cite{yuanwei}. The numbers of baseband samples for the three algorithms are set as $N_{\rm BS}/2$, $N_{\rm BS}/2$, and $N_{\rm BS}N_{\rm UE}T_s$, respectively, where $T_s$ denotes the number of ping-pong pilot iterations for the STT method}, and the near-field region is divided into the same number of grids. Particularly, the number of searching grids on the $x$-axis is set as $G_x = \eta N_{\rm BS}$, while the searching grids on the $y$-axis are set as $G_y = 20\eta$, where $\eta \in \{1,2,3\}$ denotes the oversampling rate of searching grids. In contrast, owing to the FFT implementation of the proposed algorithm, the number of searching grids on the $x$-axis is fixed as $G_x = \lceil 2(x_H-x_L)/(N_{\rm BS}\lambda) \rceil N_{\rm BS} \approx 8N_{\rm BS}$, while is identical to the baselines on the $y$-axis as $G_y = 20\eta$. As can be observed from Fig.~\ref{fig:loc_acc}(a), the cumulative distribution functions (CDFs) of localization errors $\epsilon = \sum_{\ell=0}^{L}\Vert \hat{\bf r}_{\rm S}^{(\ell)} - {\bf r}_{\rm S}^{(\ell)}  \Vert/L$ achieved by four algorithms participating in the comparison are comparable. However, the proposed time inversion algorithm entails two practical advantages over baseline schemes.

\begin{figure}[t]
	\centering
	\includegraphics[width=0.45\textwidth]{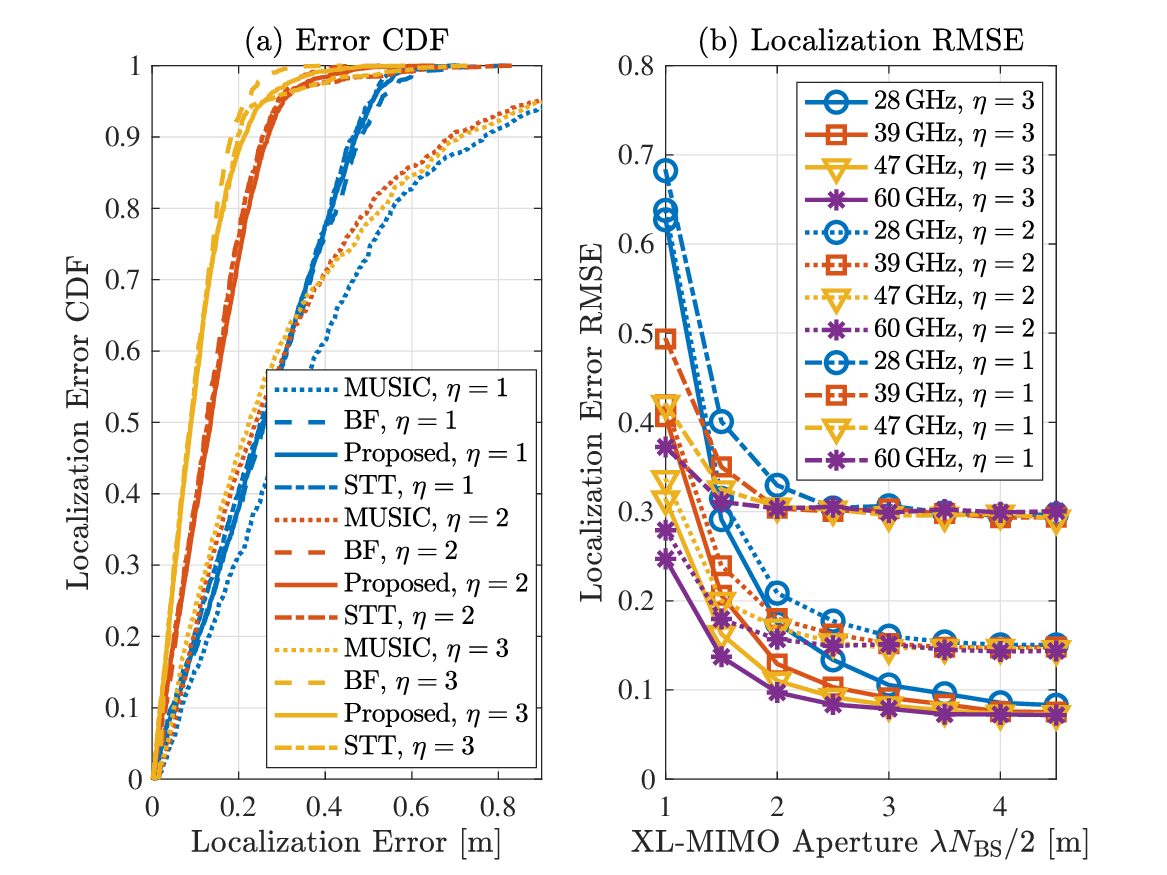}
	\caption{(a) The CDF of localization error of the considered algorithms with varying oversampling rate $\eta$ and (b) RMSE of the proposed localization algorithm with varying carrier frequencies and XL-MIMO aperture.}
	\label{fig:loc_acc}
	\vspace{-2mm}
\end{figure}

First, note that all three baseline 
algorithms require baseband samples, which incurs huge sampling overhead and power consumption in XL-MIMO systems. In contrast, the proposed time-inversion localization algorithm does not need any baseband sampling via high-frequency ADCs. Instead, as mentioned in Section~\ref{sec:arch}, acquiring power sampling via power sensors is more energy-efficient and economical. Second, the computational complexity of the proposed time inversion algorithm is far lower than that of the two baselines, as is listed in Table~\ref{tab:loc}. {For instance, according to the values of $G_x$ and $G_y$ set for different localization algorithms in this subsection, their specific computational complexity is given by $\mathcal{O}(\eta^2 N_{\rm BS}^4)$, $\mathcal{O}(\eta^2 N_{\rm BS}^2)$, $\mathcal{O}(\eta^2N_{\rm BS}^2)$, and $\mathcal{O}(\eta N_{\rm BS} \log( N_{\rm BS}))$, respectively.} It clearly shows that the proposed algorithm has the lowest complexity among the comparison algorithms even though it employs the largest number of grids on the $x$-axis, and this can also be validated by numerical results of the average runtime 
in Table~\ref{tab:loc}. In particular, the average runtime of the proposed time inversion algorithm is tens and hundreds of times lower than that of the {BF} grid search and MUSIC algorithm, respectively. Therefore, these two merits of the time inversion algorithm position it an excellent candidate for energy-efficient low-complexity localization algorithm for near-field XL-MIMO systems.

Fig.~\ref{fig:loc_acc}(b) plots the RMSE of our proposed localization algorithm with varying carrier frequencies standardized by 3GPP~\cite{3gpp.38.101}. As can be observed, a higher frequency results in a more accurate localization, and the localization error asymptotically converges as the number of XL-MIMO antenna elements $N_{\rm BS}$ increases. It was proved in Propositions~\ref{prop:one_max} and~\ref{prop:RiemannLebesgue} that as $\kappa\to\infty$ and $N_{\rm BS}\to\infty$, the peaks of $\tilde{E}({\bf r})$ will coincide the locations of the sources, i.e., ${\bf r} = {\bf r}_{\rm S}^{(\ell)}$, and the interference terms asymptotically vanish elsewhere. Therefore, Fig.~\ref{fig:loc_acc}(b) validates that the localization accuracy is improved with a smaller interference term and more prominent peaks induced by a higher carrier frequency and proliferation of antennas. 

\subsection{Sensing-Enhanced CE}
\label{sec:results_CE}
\begin{figure}[t]
	\centering
	\includegraphics[width=0.3875\textwidth]{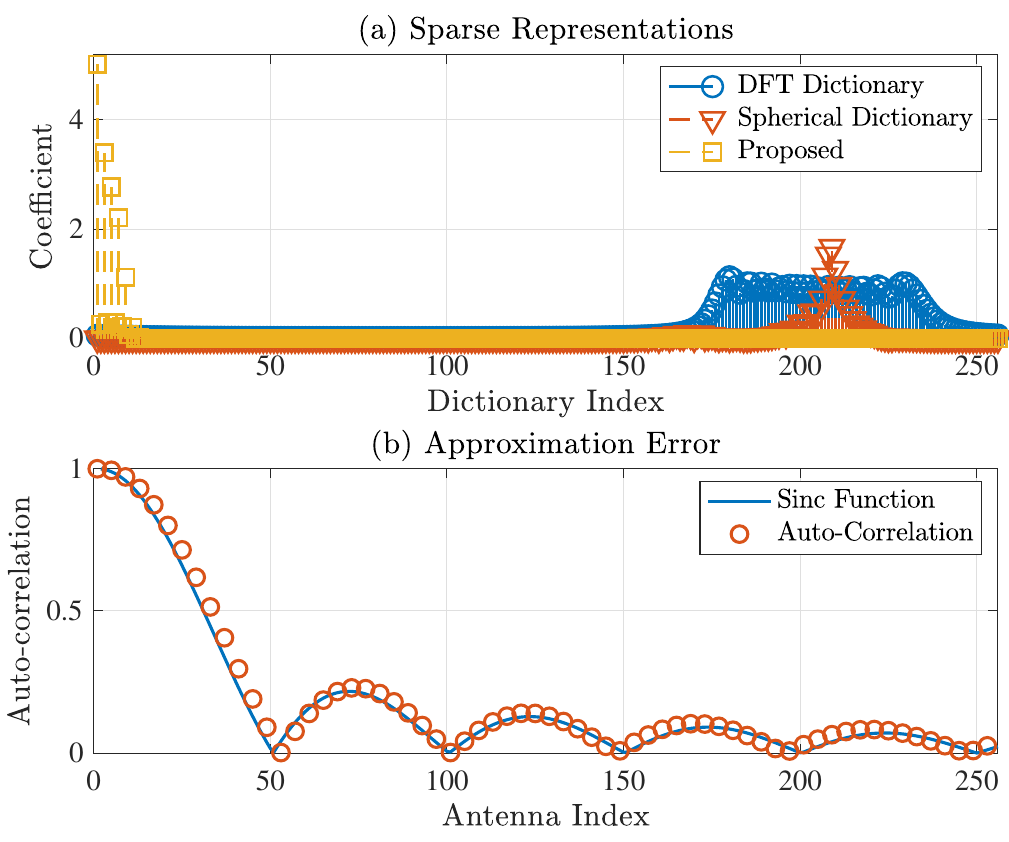}
	\caption{(a) The sparse representations of near-field channels under different dictionaries and (b) the approximation error in~\eqref{eq:toeplitzmat}.}
	\label{fig:res1}
\end{figure}
For CE performance comparison, we mainly consider two benchmark dictionaries in the simulation as follows:
\begin{itemize}
	\item \textbf{DFT dictionary}: The DFT dictionary is designed according to~\cite{6717211}, where the numbers of angle grids are set as $\beta N_{\rm BS}$ and $\beta N_{\rm UE}$ at the BS and UE side, respectively, with $\beta$ being the oversampling rate of the dictionary. Hence, the size of the DFT dictionary is $\beta^2N_{\rm BS}N_{\rm UE}$. 
	\item \textbf{Spherical wave dictionary}: The numbers of grids for distance and angle are set as $\nint{\beta\sqrt{N_{\rm BS}}}$ and $\nint{\beta\sqrt{N_{\rm UE}}}$ at the BS and UE, respectively. The angle grids are selected evenly from $-\pi/2$ to $\pi/2$, {while the reciprocal of distance grids at angle $\theta$ are evenly selected from $[(1-\cos^2\theta)/\sqrt{x_H^2+y_H^2},(1-\cos^2\theta)/y_L]$~\cite{9693928}. Therefore, the size of the spherical wave dictionary is $(\nint{\beta\sqrt{N_{\rm BS}}} \nint{\beta\sqrt{N_{\rm UE}}})^2\approx \beta^4 N_{\rm BS}N_{\rm UE}$.}
\end{itemize}

In contrast, the size of the proposed DPSS-based dictionary derived in Section~\ref{sec:ce} is $N_{\rm BS} N_{\rm UE}+L$. {Additionally, the compression ratio (CR) of CS-based CE problem is defined as $\mu = \tau/(N_{\rm UE}N_{\rm BS})$ and therefore the required number of baseband samples is $\nint{\tau N_{\rm BS}^{\rm RF}}=\nint{\mu N_{\rm BS}^{\rm RF}N_{\rm UE}N_{\rm BS}}$.} For fair comparison, the performance of CE achieved by all dictionaries is evaluated based on the OMP algorithm.

We firstly investigate the sparsification performance of the proposed dictionary. The channel sparse representations of the DFT dictionary, spherical wave dictionary, and proposed dictionary are compared in Fig.~\ref{fig:res1}(a), where the sparse representation is obtained by $\tilde{\bf h} = \boldsymbol{\Psi}^\dagger {\bf h}$. As can be observed, the conventional DFT dictionary shows a severe energy leakage effect in the near-field region~\cite{9693928}, which can be improved by the spherical wave dictionary sampled in the polar-domain. Meanwhile, the proposed DPSS-based eigen-dictionary entails the sparsest pattern among the three dictionaries. Note that the proposed dictionary is compensated by the matrix ${\bf D}_{{\rm BS}({\rm UE})}$ in~\eqref{eq:compensation} and therefore, the sparse representation contains no specific angular information. In addition, different from the DFT and spherical dictionaries, the non-zero support set appears in the first several indices since the SVD always sorts the non-zero singular values first. 
We further validate the approximation error of the derivation procedure in \eqref{eq:toeplitzmat}. As is depicted in Fig.~\ref{fig:res1}(b), the auto-correlation curve stands for the absolute value of $\overline{\bf R}_{\rm BS}[1,:]$, while the red circles show the value of the normalized sinc function. The approximation procedure shows negligible error, which confirms the high accuracy of our proposed approximation $(b)$ in~\eqref{eq:toeplitzmat}.

\begin{figure}[t]
	\centering
	\includegraphics[width=0.4\textwidth]{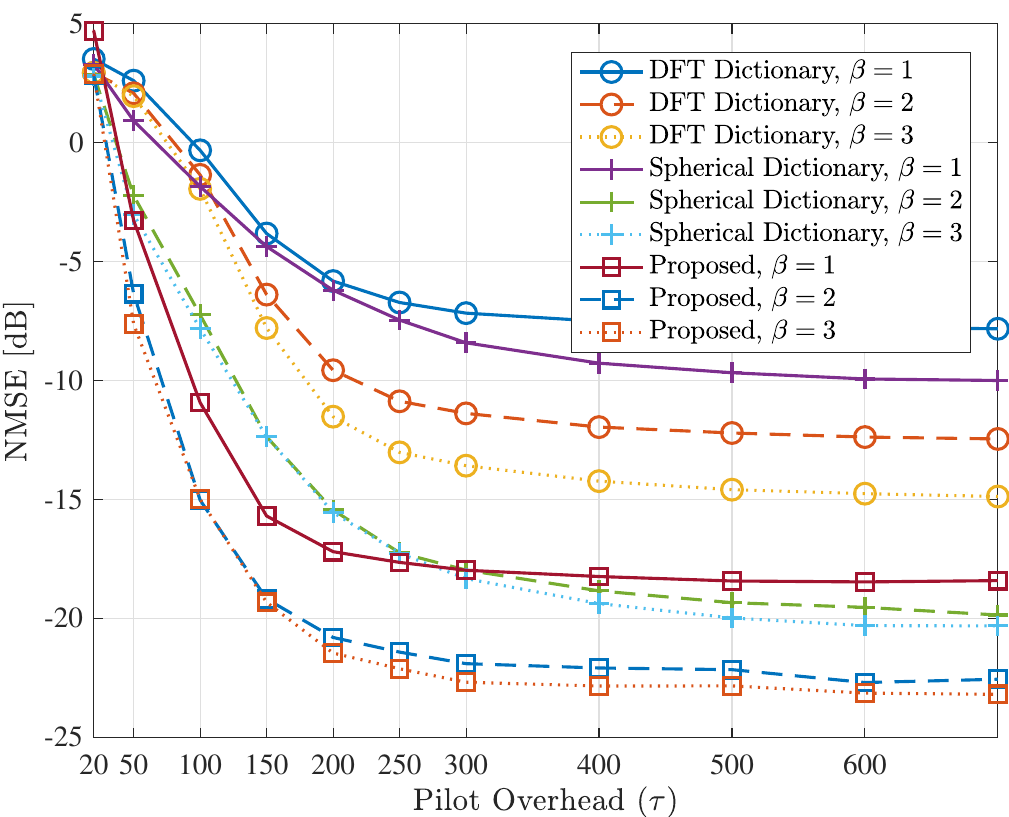}
	\caption{Performance comparison of near-field CE error versus pilot length $\tau$ and dictionary oversampling rate $\beta$ when $I = 30$.}
	\label{fig:res_overhead_1}
\end{figure}
\begin{figure}[t]
	\centering
	\includegraphics[width=0.4\textwidth]{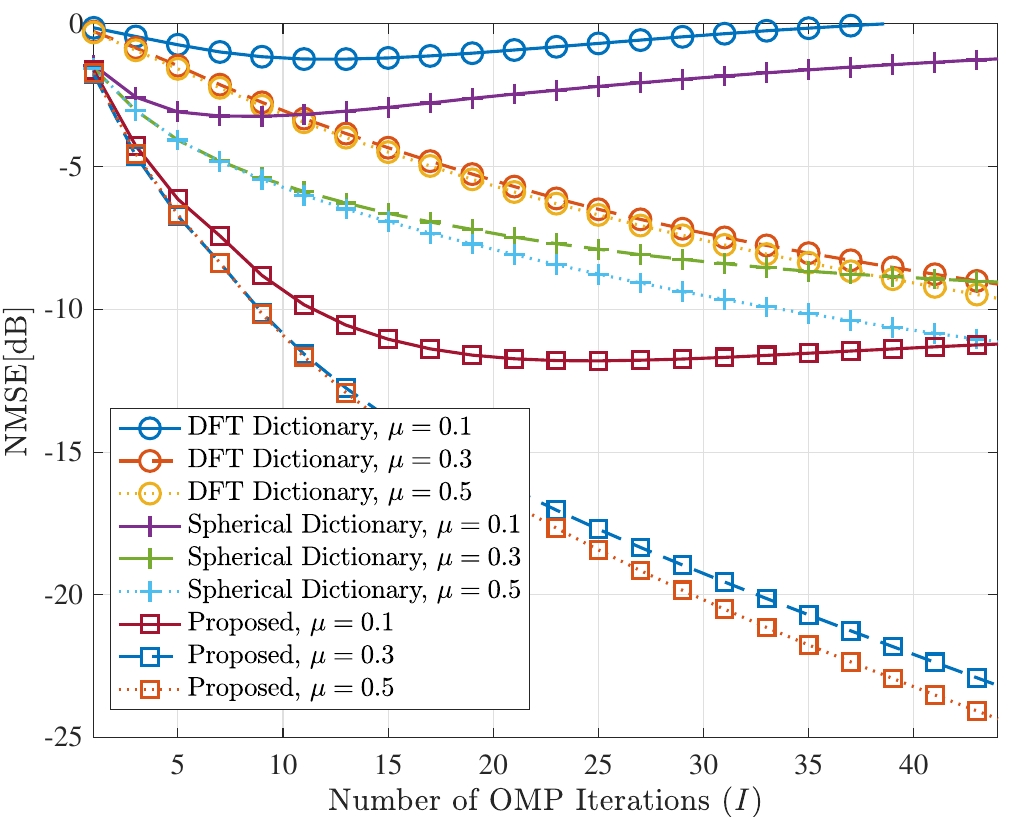}
	\caption{Performance comparison of near-field CE error versus the number of iterations $I$ and CR $\mu$ when $\beta=1$.}
	\label{fig:res2}
\end{figure}

{	
We then investigate the CE accuracy performance with pilot overhead $\tau$. As shown in Fig.~\ref{fig:res_overhead_1}, the NMSE performance of CE converges as the pilot length $\tau$ increases and the proposed schemes achieve a lower estimation error than baseline algorithms. Besides, the proposed method converges at around $\tau=300$, while the other baselines converge at around $\tau=500$. Furthermore, to achieve a CE NMSE of $-10\,{\rm dB}$, the proposed method saves up to $85.7\%$ of the pilot overhead when compared to the spherical dictionary with $\beta=1$. In addition, to achieve a CE error of $-15\,{\rm dB}$, the proposed method saves up to $50\%$ and $85.7\%$ of the pilot overhead compared to DFT and spherical dictionaries, respectively, with $\beta=3$. These results clearly demonstrate that the proposed schemes not only efficiently reduce the pilot overhead but also achieve the best estimation error with the aid of the sensing phase.

We next investigate the CE performance against the number of iterations $I$. With the reference of Fig.~\ref{fig:res_overhead_1} and recall that $\mu = \tau/N_{\rm BS}N_{\rm UE}$, we set $\mu \in \{0.1, 0.3, 0.5\}$ for further comparisons.} We firstly set the oversampling rate $\beta$ to $1$ to keep the sizes of the three considered dictionaries identical. 
As shown in Fig.~\ref{fig:res2}, both the schemes with the proposed and the spherical wave dictionaries start from lower estimation errors than the DFT dictionary, which validates the alignment of the proposed eigen-dictionary and the near-field spherical wave model. 
{For cases with $\mu = 0.1$, the under-determined observation causes severe upward warping phenomenon as shown at the right side of the curves, while the proposed method still achieves satisfying NMSE performance. For cases with $\mu>0.1$,} 
the convergence of the spherical wave dictionary experiences a deceleration due to the insufficient distance sampling interval, while the other dictionaries converge continuously. Additionally, thanks to the accurate localization information incorporated in the DPSS-based dictionary, the proposed CE scheme shows the highest estimation accuracy. 
Recall that the number of baseband samples is equal to $\nint{\mu N_{\rm BS}^{\rm RF}N_{\rm UE}N_{\rm BS}}$, which is proportional to $\mu$.
Hence, this observation implies that our proposal is more efficient in CE when utilizing the same number of ADC resources, e.g., $307$ samples when $\mu=0.3$, at the baseband. 
Finally, the achievable NMSE of the proposed method is almost identical with varying CRs $\mu$.
This phenomenon indicates that the proposed sensing-enhanced CE has great potential to alleviate the burden of baseband sampling in CE. In other words, in practical systems, only a small number of baseband samples, i.e., $0.1N_{\rm BS}N_{\rm UE}$, is required to achieve a satisfactory performance in terms of NMSE.

\begin{figure}[t]
	\centering
	\includegraphics[width=0.4\textwidth]{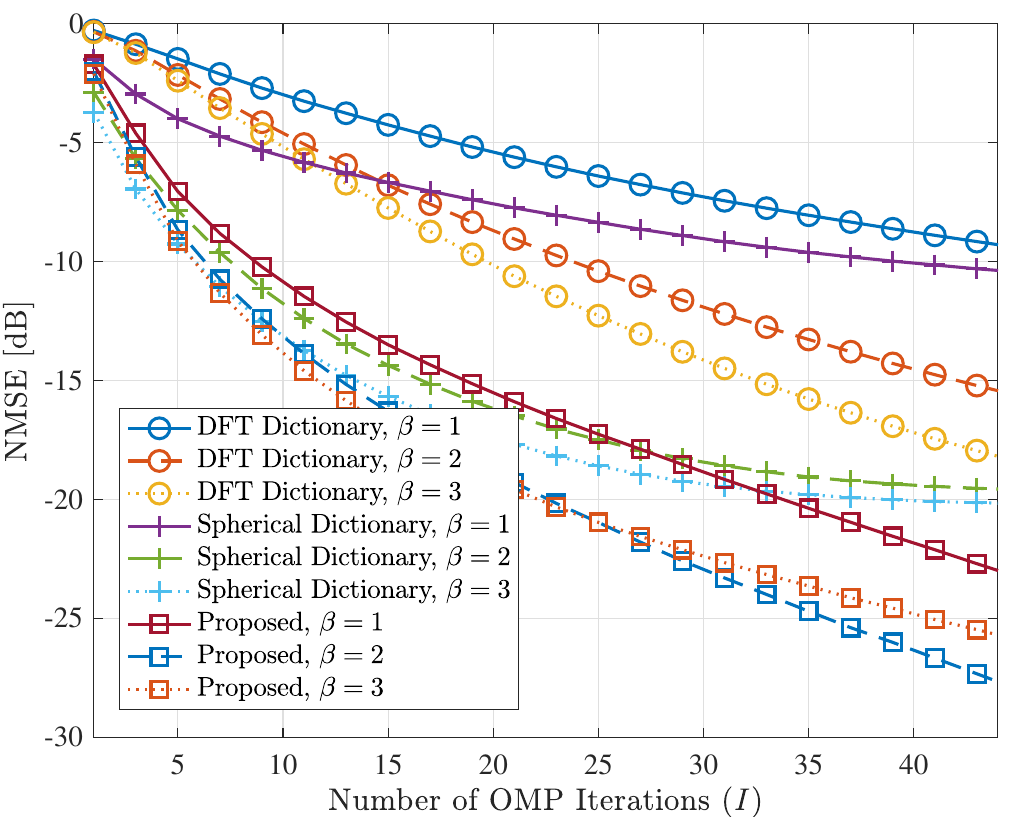}
	\caption{Performance comparison of near-field CE error versus the number of iterations $I$ and dictionary oversampling rate $\beta$ when $\mu = 0.4$.}
	\label{fig:res3}
\end{figure}
\begin{figure}[t]
	\centering
	\includegraphics[width=0.4\textwidth]{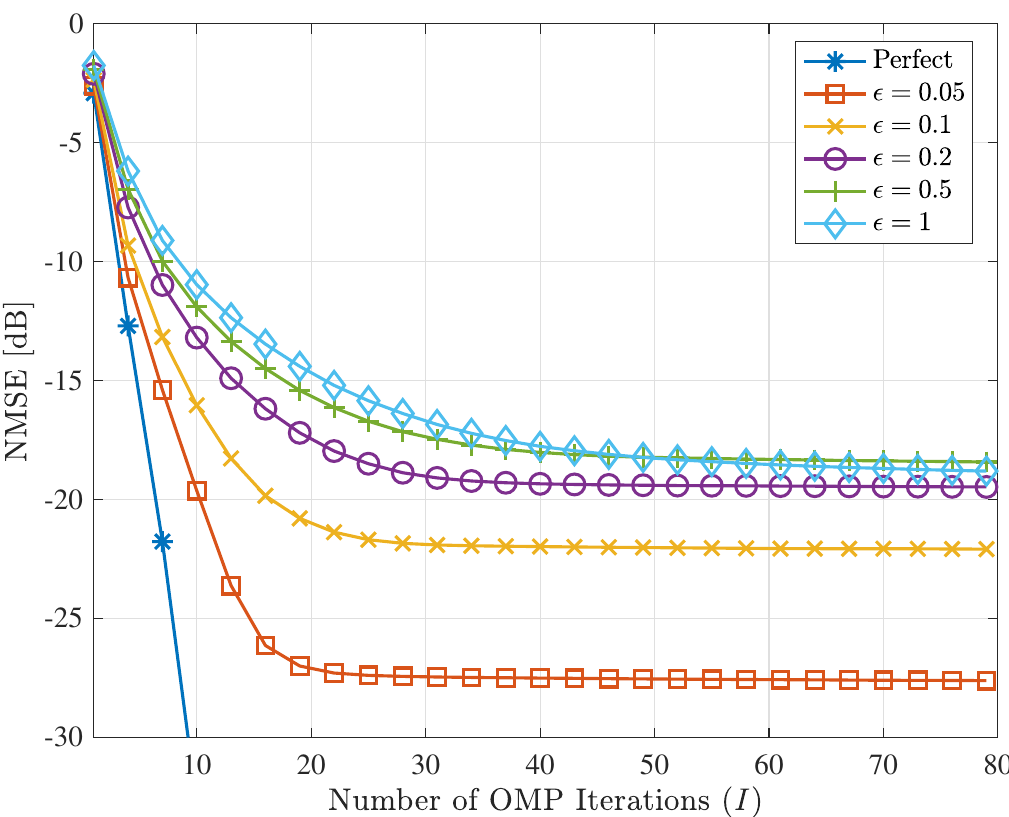}
	\caption{Robustness performance of proposed sensing-enhanced CE scheme with the DPSS-based eigen-dictionary versus different localization errors $\epsilon$ (m).}
	\label{fig:res4}
\end{figure}

We also evaluate the CE performance with oversampling rates $\beta\in\{1,2,3\}$. The CR is set to $\mu=0.4$ to maintain the same number of baseband samples. As shown in Fig.~\ref{fig:res3}, we observe significant performance improvement for all schemes by increasing $\beta$, while the proposed DPSS-based dictionary still achieves the highest reconstruction accuracy within sufficient iterations. 
Recall that the sizes of the DFT and spherical dictionaries scale up with $\beta$ and therefore more codewords can be included in these two dictionaries for enhancing the CE performance. 
In contrast, the increase in the oversampling rate $\beta$ only affects the localization accuracy in the sensing stage of our proposed CE scheme while the size of the DPSS-based dictionary remains unchanged as $N_{\rm BS}N_{\rm UE}+L$. In other words, thanks to the accurate localization information involved in the DPSS-based dictionary, the  proposed sensing-enhanced CE scheme tremendously outperforms those of two baselines even with a much smaller dictionary size.

We further investigate the robustness of the proposed dictionary versus different localization errors $\epsilon$ in Fig.~\ref{fig:res4}. The estimated coordinates are assumed to be uniformly distributed within a circular area of radius $\epsilon$. The values of $\epsilon$ in Fig.~\ref{fig:res4} are selected according to the typical localization error shown in Fig.~\ref{fig:loc_acc}(a) and the recent near-field localization field-test~\cite{9707730}. As can be observed, the proposed DPSS-based dictionary ensures the convergence of the OMP algorithm within the considered error levels up to $\epsilon=0.5\,{\rm m}$, while more OMP iterations are required for a larger value of $\epsilon$. Besides, the converged NMSE performance degrades as $\epsilon$ increases, while the converged NMSE still achieves a considerable level more accurate than $-15\,{\rm dB}$. On the one hand, this result demonstrates the robustness of the proposed CE scheme against the localization error. On the other hand, it also confirms the importance of accurate localization in near-field CE, which justifies our motivation of the sensing-enhanced CE proposal.

\subsection{Storage and Baseband Sample Analysis}
\begin{table}[t]
	\renewcommand{\arraystretch}{1.05}
	\caption{Minimum Required Baseband Samples and Dictionary Size for Target NMSE.}\label{tab:combined}
	\setlength{\tabcolsep}{2pt}
	\centering
	\begin{tabular}{c|cccccc}
		\hline\hline
		\multirow{3}{*}{\bf Dictionaries} & \multicolumn{3}{c|}{\bf Baseband Samples}                                          & \multicolumn{3}{c}{\bf Dictionary Size}                       \\ \cline{2-7} 
		& \multicolumn{6}{c}{\bf \!\!\!\!\!\!\!\!Target NMSE}                                                                                                            \\ \cline{2-7} 
		& \multicolumn{1}{c}{$\mhyphen 15\:$dB} & \multicolumn{1}{c}{$\mhyphen 20\:$dB} & \multicolumn{1}{c|}{$\mhyphen 25\:$dB} & \multicolumn{1}{c}{$\mhyphen 15\:$dB} & \multicolumn{1}{c}{$\mhyphen 20\:$dB} & $\mhyphen 25\:$dB \\ \hline
		DFT                           & \multicolumn{1}{c}{$256$}   & \multicolumn{1}{c}{$358$}   & \multicolumn{1}{c|}{$410$}   & \multicolumn{1}{c}{$2,\!304$}   & \multicolumn{1}{c}{$4,\!096$}   & $4,\!096$   \\ \hline
		Spherical                     & \multicolumn{1}{c}{$205$}   & \multicolumn{1}{c}{$358$}   & \multicolumn{1}{c|}{N/A}   & \multicolumn{1}{c}{$1,\!620$}   & \multicolumn{1}{c}{$4,\!096$}   & N/A     \\ \hline
		Proposed                     & \multicolumn{1}{c}{$\bf 154$}   & \multicolumn{1}{c}{$\bf 174$}   & \multicolumn{1}{c|}{$\bf 205$}   & \multicolumn{1}{c}{$1,\!024\!+\!L$}   & \multicolumn{1}{c}{$1,\!024\!+\!L$}   & $1,\!024\!+\!L$    \\ \hline\hline
	\end{tabular}
\end{table}
In Table~\ref{tab:combined}, we compare the minimum required baseband samples and dictionary size, i.e., the number of codewords, to achieve the NMSE targets $\{-15, -20, -25\}\,{\rm dB}$. For the former comparison we fix the oversampling rate as $\beta=2$, while for the latter comparison we set CR $\mu = 0.4$, i.e., $410$ baseband samples. As can be observed, the required baseband samples keep increasing with the target NMSE, while the proposed DPSS-based dictionary achieves the fewest required baseband samples. Besides, the sizes of the DFT and spherical dictionaries continue to increase with higher NMSE requirements, while the size of the proposed dictionary remains constant. Additionally, the $-25\,{\rm dB}$ NMSE cannot be achieved either by increasing the baseband samples or by enlarging the sizes of the spherical wave dictionary, and the corresponding entries are denoted as ${\rm N/A}$. 

Although the spherical wave dictionary achieves faster convergence during the initial iterations, the convergence performance is not as good as the other dictionaries due to the insufficient distance sampling interval and mutual correlation among codewords. 
Besides, the two DoFs in both distance and angle of the spherical wave dictionary dramatically add to the dictionary size as the resolution requirement increases. Therefore, the spherical wave dictionary requires the largest number of baseband samples and codewords to reach the target NMSE. On the contrary, thanks to the mutual orthogonality among codewords, the DFT dictionary can satisfy more stringent NMSE requirements with smaller sizes than the spherical ones. Yet, its mismatch with the near-field channel model still leads to a bulkier dictionary compared to the proposed DPSS-based one.

In contrast, thanks to the prior localization information integrated in the proposed dictionary, it is able to reach the NMSE targets with a drastically smaller number of baseband samples. 
For example, to achieve a $-20\,$dB NMSE, the required baseband samples show a $51\%$ reduction compared to the DFT and spherical dictionaries, respectively.
Furthermore, compared to the DFT and spherical wave dictionaries, the proposed DPSS-based dictionary does not need to sacrifice the NMSE performance for a lower storage, which is exemplified by a $66\%$ reduction in dictionary size for a given $-20\,$dB target NMSE.

\section{Conclusion}

In this paper, we introduced a cost-effective XL-MIMO transceiver architecture measuring the power of the incident wavefront, based on which a time inversion near-field localization algorithm was proposed. The proposed localization scheme achieves accurate localization performance compared to the widely-adopted algorithms without baseband sampling through power-hungry ADCs, and is computationally efficient thanks to the FFT implementation. To fully exploit the estimated locations, we further proposed a novel DPSS-based eigen-dictionary for near-field XL-MIMO CE. By leveraging the EVD associated with the near-field channel, the proposed DPSS-based dictionary outperforms the conventional DFT and polar-domain spherical wave dictionaries in channel sparsification, NMSE in CE, storage requirements, and more importantly, the required baseband samples. The proposed sensing-enhanced CE approach presents an ideal solution for a low-complexity yet efficient CE scheme with minimized hardware requirements and training overhead.

\appendices
\section{Proof of Proposition \ref{prop:RiemannLebesgue}}
\label{sec:appendix2}
We firstly treat the BS array as a continuous aperture array and respectively define the distance from the $m$-th element on the UE array ${\bf r}_{\rm UE}^{(m)}$ and the target point ${\bf r}$ to the point $(t,0)$ on the BS array as
\begin{align}
	r_1(t) &= \Vert {\bf r}_{\rm UE}^{(m)}-{\bf r}_{\rm BS}^{(n)} \Vert = \sqrt{(x_{\rm UE}^{(m)}-t)^2+(y_{\rm UE}^{(m)})^2},\\
	r_2(t) &= \Vert {\bf r}-{\bf r}_{\rm BS}^{(n)} \Vert = \sqrt{(x-t)^2+y^2},
\end{align}
where $t\in{\mathcal D} = [-(N_{\rm BS}-1)\lambda/4,(N_{\rm BS}-1)\lambda/4]$, and $f(t) = r_1(t)-r_2(t)$ is continuous on ${\mathcal D}$. The desired term $\tilde{E}_\ell^{\rm d} ({\bf r})$ in~\eqref{eq:desvanish} then yields
\begin{equation}
	\begin{aligned}
		\tilde{E}_\ell^{\rm d} ({\bf r};\kappa) &= \sum_{n=1}^{N_{\rm BS}} e^{-\jmath\kappa \Vert {\bf r}_{\rm UE}^{(m)}-{\bf r}_{\rm BS}^{(n)} \Vert + \jmath\kappa \Vert {\bf r}-{\bf r}_{\rm BS}^{(n)} \Vert}\\ 
		& \overset{(c)}{\simeq}\int_{\mathcal{D}} e^{-\jmath\kappa f(t)}~{\rm d}t = \int_{\mathbb{R}} p_{x}(x) e^{-\jmath\kappa x}~{\rm d}x,
	\end{aligned}
	\label{eq:dftt}
\end{equation}
where $(c)$ asymptotically holds as $N_{\rm BS}\to\infty$, $p_x(x) = \frac{2}{(N_{\rm BS}-1)\lambda}\vert\frac{\partial f^{-1}(x)}{\partial x} \vert$ is the probability density function of $f(t)$ over $t\in\mathcal{D}$, and is continuous over $\mathbb{R}$. Substitute the integrated variable in~\eqref{eq:dftt} by $x^\prime = x+\pi/\kappa$ then yields
\begin{equation}
	\begin{aligned}
		\int_{\mathbb{R}} p_{x}(x^\prime) e^{-\jmath\kappa x^\prime}~{\rm d}x^\prime&= \int_{\mathbb{R}} p_{x}\left(x+\frac{\pi}{\kappa}\right) e^{-\jmath\kappa x}e^{-\jmath\pi}~{\rm d}x \\ 
		&= -\int_{\mathbb{R}} p_{x}\left(x+\frac{\pi}{\kappa}\right) e^{-\jmath\kappa x}~{\rm d}x.
	\end{aligned}
	\label{eq:dfttt}
\end{equation}
Taking the average of~\eqref{eq:dftt} and~\eqref{eq:dfttt}, we obtain
\begin{equation}
	\begin{aligned}
		\left\vert \tilde{E}_\ell^{\rm d} ({\bf r};\kappa) \right\vert & \overset{(c)}{\simeq} \left\vert \frac{1}{2}\int_{\mathbb{R}} \left( p_x(x)- p_{x}\left(x+\frac{\pi}{\kappa}\right)  \right)e^{-\jmath\kappa x}~{\rm d}x \right\vert\\
		&\leq\frac{1}{2}\int_{\mathbb{R}} \left\vert  p_{x}(x) - p_{x}\left(x+\frac{\pi}{\kappa}\right)  \right\vert~{\rm d}x.
	\end{aligned}
\end{equation}
Given that $p_x(x)$ is continuous, $\vert  p_{x}(x) - p_{x}\left(x+\frac{\pi}{\kappa}\right) \vert$ will converge to $0$ as $\vert\kappa\vert\rightarrow\infty$ for all $r\in \mathbb{R}$, therefore $\vert \tilde{E}_\ell^{\rm d} ({\bf r};\kappa)\vert$ will vanish to $0$.

The interference term can be formulated as
\begin{equation}
	\begin{aligned}
		\tilde{E}_m^{{\rm i}}({\bf r}) &=\sum_{n=1}^{N_{\rm BS}}\frac{e^{-2\jmath \xi_n + \jmath\kappa\left( \Vert {\bf r}_{\rm S}^{(\ell)}-{\bf r}_{\rm BS}^{(n)} \Vert+ \Vert {\bf r}-{\bf r}_{\rm BS}^{(n)} \Vert\right) }}{d_0^{(\ell,n)}}\\
		&=\sum_{n=1}^{N_{\rm BS}}\frac{e^{ \jmath {\phi}^{(\ell,n)}_{\bf r} }}{d_0^{(\ell,n)}} e^{-2\jmath{\xi}_n },
		\label{eq:newintf}
	\end{aligned}
\end{equation}
where ${\phi}^{(\ell,n)}_{ \bf r} = \kappa ( \Vert {\bf r}_{\rm T}^{(m)}-{\bf r}_{\rm BS}^{(n)} \Vert + \Vert {\bf r}-{\bf r}_{\rm BS}^{(n)} \Vert)$ is the overall phase at the location of near-field coordinate ${\bf r}$. The expectation of $\tilde{E}_m^{{\rm i}}({\bf r})$ can then be given as
\begin{equation}
	\begin{aligned}
		&\mathbb{E}\left[\tilde{E}_m^{{\rm i}}({\bf r}) \right] \\ ={}&\mathbb{E}\left[\sum_{n=1}^{N_{\rm BS}} \left( \frac{ e^{\jmath {\phi}^{(\ell,n)}_{ {\bf r}} }}{d_0^{(\ell,n)}} \right) \left( \cos\left( 2{\xi}_n \right)+\jmath\sin \left( 2{\xi}_n \right) \right)\right].
	\end{aligned}
\end{equation}
If the phase shift of reference wave is drawn from a uniform distribution, i.e., $\xi_n\sim {\mathcal{U}}[0,\pi]$, the interference term will vanish to $0$ asymptotically, as $\mathbb{E}[\cos(2\xi_n)]=0$ and $\mathbb{E}[\sin(2\xi_n)]=0$ hold for $N_{\rm BS}\rightarrow \infty$\cite{tagoram}.

Because $\mathbf{e}_\mathrm{n}$ is an AWGN vector, ${\rm Re}\left(\mathbf{e}_\mathrm{n}[n]\right)$ is a Gaussian random variable with zero mean. When $N_\mathrm{BS}\to\infty$, the noise term $\tilde{E}_\ell^{\rm n}({\bf r}) \to 0$ according to the law of large numbers.
\qed
\bibliographystyle{IEEEtran}
\bibliography{IEEEabrv,references}

\end{document}